\documentclass{amsart}
\usepackage{amssymb}
\usepackage{amsmath}
\usepackage{amsthm}
\usepackage{amscd}
\usepackage{amstext}
\usepackage{amsfonts}
\usepackage[all]{xy}
\usepackage[dvipdfmx]{graphicx}
\usepackage{bm}

\def\rnum#1{\expandafter{\romannumeral #1}} 

\newenvironment{Proof}[2]{\removelastskip\vspace{6pt}\noindent
 {\bf Proof  #1.}~\rm#2}{\par\vspace{6pt}}
 
 \theoremstyle{plain}
\newtheorem{theorem}{Theorem}[section]
\newtheorem{proposition}[theorem]{Proposition}
\newtheorem{lemma}[theorem]{Lemma}

\newtheorem{assumption}[theorem]{Assumption}
\theoremstyle{remark}
\newtheorem{definition}[theorem]{\rm Definition}
\newtheorem{example}[theorem]{\rm Example}
\newtheorem{remark}[theorem]{\rm Remark}

\newcommand{\Z}{\mathbb{Z}}
\newcommand{\R}{\mathbb{R}}
\newcommand{\C}{\mathbb{C}}
\newcommand{\T}{\mathbb{T}}
\newcommand{\D}{\mathbb{D}}
\newcommand{\U}{\mathbf{U}}
\newcommand{\GL}{\mathbf{GL}}
\newcommand{\SSS}{\mathbb{S}}
\newcommand{\Vect}{\mathrm{Vect}}

\newcommand{\Fred}{\mathrm{Fred}}
\newcommand{\I}{\mathcal{I}}
\newcommand{\B}{\mathrm{B}}
\newcommand{\GP}{\mathrm{GP}}
\newcommand{\Bulk}{\mathrm{Bulk}}
\newcommand{\Edge}{\mathrm{Edge}}
\newcommand{\cpt}{\mathrm{cpt}}
\newcommand{\pt}{\mathrm{pt}}
\newcommand{\id}{\mathrm{id}}

\newcommand{\AII}{\mathrm{A\hspace{-.1em}I\hspace{-.1em}I}}
\newcommand{\tKSp}{\widetilde{KSp}\hspace{-.15em}\,^0}
\newcommand{\interior}[1]{%
  {\kern0pt#1}^{\mathrm{o}}%
}

\DeclareMathOperator{\qsf}{\mathrm{sf_2}}
\DeclareMathOperator{\ind}{\mathrm{ind}}
\DeclareMathOperator{\End}{\mathrm{End}}
\DeclareMathOperator{\Ker}{\mathrm{Ker}}
\DeclareMathOperator{\Coker}{\mathrm{Coker}}
\DeclareMathOperator{\rank}{\mathrm{rank}}

\begin{document}

\title[Bulk-edge correspondence and the cobordism invariance]
{Bulk-edge correspondence and the cobordism invariance of the index}
\author{Shin Hayashi}
\address{Graduate School of Mathematical Sciences, University of Tokyo, 3-8-1 Komaba,
Tokyo, 153-8914, Japan.}
\address{Mathematics for Advanced Materials-OIL, AIST-Tohoku University, Sendai,
980-8577, Japan.}
\email{{\tt shin-hayashi@aist.go.jp}}


\begin{abstract}
We show that the bulk-edge correspondence for two-dimensional type A and type $\AII$ topological insulators follows directly from the cobordism invariance of the index.
\end{abstract}

\maketitle
\setcounter{tocdepth}{1}
\tableofcontents

\section{Introduction}
\label{sec:1}
In condensed matter physics, a correspondence between two topological invariants, called the bulk-edge correspondence, is well-known.
In this paper, we study the bulk-edge correspondence from the point of view of $K$-theory and index theory.
We show that the bulk-edge correspondence for two-dimensional type A and type $\AII$ topological insulators follows directly from the cobodrism invariance of the index, which is a basic property of the index.

The bulk-edge correspondence was recognized by the theoretical study of the quantum Hall effect.
The quantum Hall effect was discovered by K. von Klitzing, G. Dorda and M. Pepper in $1980$ \cite{KDP80}.
They observed a quantization of the Hall conductance of a two-dimensional electron gas in a strong magnetic field.
Such quantization was explained by D. J. Thouless, M. Kohmoto, M. P. Nightingale and M. den Nijs from the topological point of view.
They considered such effect on an infinite system without edge, and showed that the Hall conductance of the system coincides with the first Chern number (called the TKNN number) of the Bloch bundle on the Brillouin torus by using the Kubo formula \cite{TKNN82,Ko85}.
B. I. Halperin considered such phenomena on a system with edge \cite{Hal82}. Y. Hatsugai showed that the Hall conductance of a system with edge corresponds to a winding number (or a spectral flow) counted on a Riemann surface \cite{Hat93a}.
The equality between the first Chern number defined for such infinite system without edge (called the {\em bulk index}) and the winding number defined for such system with edge (called the {\em edge index}) was proved by Hatsugai \cite{Hat93b} by using a Riemann surface.
This correspondence is called the {\em bulk-edge correspondence}.
In $2005$, C. L. Kane and E. J. Mele proposed the quantum spin Hall effect \cite{KM05a}.
They defined a $\Z_2$-valued topological invariant (called the Kane-Mele invariant) for two-dimensional systems with odd time-reversal symmetry, and distinguished the quantum spin Hall state and the insulating state \cite{KM05b}.
On the quantum spin Hall state, there is a spin current carried on the edge.
This is the bulk-edge correspondence for the quantum Hall system.
The quantum spin Hall effect was observed experimentally by M. K\"{o}nig {\em et al}. \cite{Ko07} based on the model proposed by B. A. Bernevig T. L. Hughes and S.-C. Zhang \cite{BHZ06}.

By the pioneering work of J. Bellissard, a mathematical approach to the quantum Hall effect from the point of view of the noncommutative geometry was developed \cite{Be86,BvES94}.
J. Kellendonk, T. Richter and H. Schulz-Baldes gave a proof of the bulk-edge correspondence in such a way that disordered systems were also treated.
They used $K$-theory for $C^*$-algebras, especially the boundary homomorphism of the six-term exact sequence associated to the Toeplitz extension \cite{HKR00,HKR02}.
Since they used the pairing of $K$-theory and cyclic cohomology in the definition of bulk and edge invariants, this method cannot be applied to $\Z_2$-valued invariants as the quantum spin Hall effect.
C. Bourne, A. L. Carey, Kellendonk and A. Rennie overcome such difficulty  by using Kasparov's $KK$-theory \cite{BCR1,BCR2,BKR16}.
Y.~Kubota studied topological phases and the bulk-edge correspondence from the point of view of the coarse geometry.
He used the coarse Mayer-Vietoris exact sequence to prove the bulk-edge correspondence \cite{Ku15}.
Their methods are also applied to any symmetry and disordered systems.

G. M. Graf and M. Porta gave a proof of the bulk-edge correspondence in \cite{GP13} by using another vector bundle. 
They defined a vector bundle over some torus for a system without edge by using (formal) solutions of the Hamiltonian which decays as it goes to a specific direction.
This vector bundle is closely related to the edge index.
By using this bundle, they showed the bulk-edge correspondence for the quantum spin Hall system.
They also gave a proof of the bulk-edge correspondence for the quantum Hall system.
Note that such vector bundles consisting of decaying solutions were also considered in \cite{Th83}.
Since this vector bundle is defined without using the translation invariance in the direction across the edge, disordered systems are also be treated (the translation invariance along the edge is assumed).
They called topological invariants for such disordered systems obtained by using such vector bundles bulk indices in the sense that these invariants are defined by the bulk Hamiltonian and do not use the Dirichlet boundary condition.
Graf--Porta showed that such topological invariant defined by decaying solutions  equals to the spectral flow.
For periodic systems, they also showed that such topological invariant equals to the bulk index (the Kane-Mele invariant or the TKNN number) by some consideration on a Riemann surface.
Graf--Porta's proof is based on the functional analysis and the basic homotopy theory.
Graf and Porta gave in \cite{GP13} another proof of the bulk-edge correspondence for quantum Hall systems based on Levinson's theorem in scattering theory.

There are many other works about the bulk-edge correspondence, especially from the point of view of $K$-theory and index theory.
Avron--Seiler--Simon considered an index theorem for a pair of projections in this context \cite{ASS94}, and Elbau--Graf proved the bulk-edge correspondence for the quantum Hall system by generalizing their method \cite{EG02}.
In \cite{ASV13} Avila--Schulz-Baldes--Villegas-Blas proved the bulk-edge correspondence for quantum spin Hall systems by using transfer matrix methods as in \cite{Hat93b}.
Thiang--Varghese used T-duality to simplify the bulk-edge correspondence. The case of systems with symmetry are also treated in their paper with K. C. Hannabuss \cite{MT15,MT16a,MT16b,HMT16}.
In \cite{KLW15b}, Kaufmann--Li--Wehefritz-Kaufmann studied the $\Z_2$-valued invariant from the point of view of $K$-theory and index theory.
In \cite{KLW15b}, a proof of the bulk-edge correspondence for topological insulators is given by using the Baum--Connes isomorphism.
Note that $K$-theory is also used in the classification of topological phases \cite{Kit09,FM13,Kel15,Thi16,Ku16}.

In this paper, we study Graf--Porta's proof of the bulk-edge correspondence to use another vector bundle, from the point of view of $K$-theory and index theory.
We give a proof of the bulk-edge correspondence for some two-dimensional type A and type $\AII$ topological insulators in the Altland-Zirnbauer classification \cite{AZ97}.
Our proof is based on the cobordism invariance of the index.
The proof goes as follows. We first construct two elements of some compactly supported $K$-groups, and see that their family indices are the $K$-class of the Bloch bundle (Lemma \ref{lem4}) and the $K$-class of the difference of the class of vector bundle consisting of decaying solutions and that of a trivial bundle (Lemma \ref{lem5}), respectively.
We then construct a cobordism between them.
The cobordism is given by an element of a compactly supported $K$-group whose support is a Riemann surface considered in \cite{GP13}.
By using the cobordism invariance of the index, we obtain a relation between the bulk index and the first Chern number of the vector bundle consisting of decaying solutions.
We next show a relation between the latter invariant and the edge index by using the localization of $K$-class and the excision property of the index.
This part is just a computation of some element of a $K$-group.

There are two distinct approaches known for the bulk-edge correspondence: one is based on functional analysis (\cite{GP13}, for example) and the other is based on $K$-theory and noncommutative geometry (\cite{HKR02}, for example).
Historically speaking, M.~F. Atiyah and I.~M. Singer's index theorem \cite{AS1} was stated by using the framework of topological $K$-theory and gave one perspective to previously known topological invariants, and noncommutative geometry widely extends this perspective.
In this paper, we study Graf-Porta's functional analytic proof from the viewpoint of the topological $K$-theory and the classical Atiyah-Singer's index theorem.
In this way, we can reconcile these two approaches.

\section{Preliminaries}
\label{sec:2}
In this section, we collects some basic facts about topological $K$-theory (complex $K$-theory and $KSp$-theory) needed in this paper.
A definition of $\Z$-valued spectral flow and $\Z_2$-valued spectral flow is contained.
For a precise definition of the spectral flow, we find it convenient to follow the definition given by J. Phillips \cite{Ph96}.
For the later use, we give here an explicit construction of the isomorphism $\tKSp(\SSS^{2,1}) \cong \Z_2$.
\subsection{Topological $K$-theory}
\label{sec:2.1}
\subsubsection{Complex $K$-theory}
\label{sec:2.1.1}
Here we collect some basic facts about complex $K$-theory.
Details can be found in \cite{At67,AS1,Se68,Ka78,LM89}.
Let $X$ be a compact Hausdorff space. The isomorphism classes of finite rank complex vector bundles over $X$ makes an abelian semigroup by taking direct sum as its binary operation. We define the $K$-group $K^0(X)$ for $X$ as the Grothendieck group of this abelian semigroup. We denote the class of a vector bundle $E$ by $[E]$. $K^0(X)$ is a commutative ring with unit by taking tensor product as its multiplication.
For a finite dimensional complex vector space $V$, we write $\underline{V}$ for the product bundle $X \times V$ over $X$.\footnote{Note that the notation $\underline{V}$ can mean product bundles over different base spaces. We use this notation for simplicity. Its base space will be clear from the context.}
Let $x$ be a point of $X$. We define the reduced $K$-group for a based space $(X,x)$ by $\widetilde{K}^0(X) := \Ker(i^* \colon K^0(X) \rightarrow K^0(\{ x \}))$, where $i \colon \{ x \} \hookrightarrow X$ is a natural inclusion.
Let $A$ be a closed subspace of $X$. We define the $K$-group for a pair $(X,A)$ by $K^0(X, A) := \widetilde{K}^0(X/A)$, where we take $A/A$ as a base point of $X/A$.
For a locally compact Hausdorff space $Y$, we take its one-point compactification $Y^+$. Then $(Y^+, +)$ is a based compact Hausdorff space. We define $K^0_\cpt(Y) := \widetilde{K}^0(Y^+)$, which is called the compactly supported $K$-group for $Y$.

Let $n$ be a positive integer. We define $K^{-n}(X, A) := \widetilde{K}^0(\sum^n(X/A))$ where $\sum^n(X/A)$ is the $n$-fold reduced suspension.
Let $\emptyset$ be the empty set, then we have $K^{-n}(X) = K^{-n}(X, \emptyset)$.
We have the following long exact sequence.
\begin{equation*}
	\cdots \rightarrow K^{-1}(X, A) \rightarrow K^{-1}(X) \rightarrow K^{-1}(A) \overset{\partial}{\rightarrow} K^0(X, A) \rightarrow K^0(X) \rightarrow K^0(A).
\end{equation*}

There is an alternative description of $K$-groups. $K^0_\cpt(Y)$ can be defined as an equivalence classes of the isomorphism classes of triples $(E, F ; f)$, where $E$ and $F$ are finite rank complex vector bundles over $Y$ and $f \colon E \rightarrow F$ is a bundle homomorphism invertible outside a compact set. Its equivalence relation is generated by stabilization and homotopy. We denote its class by $[E, F; f]$.

Let $\GL(\infty, \C)$ be the inductive limit of a sequence $\GL(1, \C) \rightarrow \GL(2, \C) \rightarrow \cdots$, where $\GL(n, \C) \rightarrow \GL(n+1,\C)$ is given by $\mathbf{A} \mapsto \mathrm{diag}(\mathbf{A}, \bm{1})$. Then we have a natural isomorphism $[X, \GL(\infty, \C)] \cong K^{-1}(X)$.
We denote by $[f]$ the homotopy class of a continuous map $f \colon X \rightarrow \GL(\infty, \C)$.

Let $H$ be a separable complex Hilbert space. Let the space $\mathrm{Fred}(H)$ the space of bounded linear Fredholm operators on $H$ with norm topology.
Then there is a bijection $\mathop{\mathrm{index}} \colon [X, \mathrm{Fred}(H)] \rightarrow K^0(X)$ given by taking the family index, which was shown by Atiyah and K. J{\"a}nich independently.
We refer the reader to \cite{At67} for the construction of this map.
Instead of explaining the construction, we note that, if a continuous map $f \colon X \rightarrow \mathrm{Fred}(H)$ consists of a family $f(x)$ of Fredholm operators whose dimension of kernels are constant, then $\mathop{\mathrm{index}}(f) = [\Ker(f)] - [\Coker(f)]$, where $\Ker(f)$ and $\Coker(f)$ are vector bundles over $X$ whose fibers at a point $x$ in $X$ are $\Ker(f(x))$ and $\Coker(f(x))$ respectively.

Let $E$ be an even dimensional spin$^c$ vector bundle on $Y$. Then we have an isomorphism $K^0_\cpt(Y) \rightarrow K^0_\cpt(E)$, called the {\em Thom isomorphism}, given by taking a cup product with the Thom class associated to the spin$^c$ structure of $E$.
If $E = Y \times \C$ and if we consider for $E$ a spin$^c$ structure naturally induced by the complex structure of $E$, the Thom isomorphism (or the Bott isomorphism) $K^0_\cpt(Y) \rightarrow K^0_\cpt(Y \times \C)$ is given by taking a cup product with the Bott class $[\underline{\C}, \underline{\C}; z] \in K^0_\cpt(\C)$. 

Let $X$ and $Y$ be manifolds possibly with boundary and let $f \colon X \rightarrow Y$ be a neat embedding.\footnote{An embedding $f \colon X \hookrightarrow Y$ is said to be {\em neat} if $X \cap \partial Y = \partial X$ and the space $T_xX$ is not contained in $T_x(\partial Y)$ for any point $x \in \partial X$ (see \cite{Hir94}).}
Let $N$ be the normal bundle of this embedding. We assume that $N$ is even dimensional and equipped with a spin$^c$ structure.
We take a tubular neighborhood of the embedded manifold, and identify it with $N$.
Then we can define a {\em push-forward map} $f_! \colon K^0_\cpt(X) \rightarrow K^0_\cpt(Y)$ by the composition of the Thom isomorphism $K^0_\cpt(X) \rightarrow K^0_\cpt(N)$ and the map $K^0_\cpt(N) \rightarrow K^0_\cpt(Y)$ induced by the collapsing map $Y^+ \rightarrow Y^+/(Y^+ - N) \cong N^+$. The push-forward map $f_!$ is independent of the choice of a tubular neighborhood.

Let $X$ be an even-dimensional spin$^c$ manifold without boundary. We take an embedding $j$ of $X$ into an even dimensional Euclidean space $\R^{2n}$.
We fix a spin$^c$ structure on $\R^{2n}$.
Then the normal bundle $N$ of this embedding has a naturally induced spin$^c$ structure, and we have a push-forward map $j_! \colon K^0_\cpt(X) \rightarrow K^0_\cpt(\R^{2n})$.
We take an orientation preserving linear isomorphism $\R^{2n} \cong \C^n$, and consider the inverse of Thom isomorphism $\beta^{-1} \colon K^0_\cpt(\R^{2n}) \cong K^0_\cpt(\C^n) \rightarrow K^0(\{ \pt\}) \cong \Z$.

\begin{definition}
We define a homomorphism $\ind_{X} \colon K^0(X) \rightarrow \Z$ by $[E] \mapsto \beta^{-1} \circ j_! ([E] \bm{\otimes} [L(X)])$, where $L(X)$ is the determinant line bundle associated to the spin$^c$ structure of $X$.\footnote{The determinant line bundle $L(X)$ associates to the principal spin$^c$ bundle of the spin$^c$ structure by the homomorphism ${\bf Spin^c}(2n) = ({\bf Spin}(2n) \times {\bf U}(1))/\{ \pm 1 \} \ni [\lambda, z] \mapsto z^2 \in {\textbf U}(1)$.}
\end{definition}

\begin{remark}\label{rem1}
The map $\ind_X$ coincides with the composite of
the Thom isomorphism $K^0(X) \rightarrow K^0_\cpt(TX)$ associated to the spin$^c$ structure of $X$, and the {\em topological index} $K^0_\cpt(TX) \rightarrow \Z$ defined by Atiyah and Singer in \cite{AS1}.\footnote{There exists different sign conventions in $K$-theory. If we use another conventions, the expression of $\ind_X$ will be different, for example. In this paper, we follow the one used in \cite{AS1}.}
\end{remark}

\subsubsection{$KSp$-theory}
We here collect a definition and basic properties of $KSp$-theory.
The definition of Real vector bundle is also included for the later use.
We refer the reader to \cite{At66,Du69,NG2,NG1} for the details.
Let $X$ be a topological space.
A homeomorphism $\tau_X \colon X \rightarrow X$ of period $2$ (i.e. $\tau_X^2 = \id_X$) is called an involution on $X$.
We call the pair $(X, \tau_X)$ an {\em involutive space}.
Let $X^\tau$ be the fixed point set of this involution.

\begin{example}
Let $\R^{p,q}$ be the involutive space $\R^q \oplus i\R^p$ whose involution\footnote{
Note that there are two different but well-used notations for this space.
Our involutive space $\R^{p,q}$ is denoted by ``$\R^{q,p}$'' in \cite{NG2,NG1}, for example.
This leads to different notations for the following involutive spaces $\D^{p,q}$ and $\SSS^{p,q}$.
We here follow Atiyah's original notation in \cite{At66}.
}
 is given by $\id$ on the first $\R^q$ and $-\id$ on the second $\R^p$.
We also denote by $\D^{p,q}$ the unit ball in $\R^{p,q}$ and $\SSS^{p,q}$ for the unit sphere in $\R^{p,q}$ with induced involutions.
Note that the involutive space $\SSS^{p,q}$ is topologically $\SSS^{p+q-1}$.
\end{example}

\begin{definition}
A pair $(E, \Theta)$ is called a {\em Quaternionic} (resp. {\em Real}) {\em vector bundle} over the involutive space $(X, \tau_X)$,
if $E$ is a complex vector bundle over the space $X$ and $\Theta \colon E \rightarrow E$ is a (topological) homeomorphism which satisfies,
(\rnum{1}) the projection $E \rightarrow X$ is $\Z_2$-equivariant,
(\rnum{2}) $\Theta \colon E_x \rightarrow E_{\tau_X(x)}$ is anti-linear and
(\rnum{3}) $\Theta^2 = -\id_E$ (resp. $\Theta^2 = \id_E$).
\end{definition}
A homomorphism between Quaternionic (resp. Real) vector bundles $(E, \Theta_E)$ and $(F, \Theta_F)$ over $(X, \tau_X)$ is
a homomorphism $\varphi \colon E \rightarrow F$ of complex vector bundles which satisfies $\varphi \circ \Theta_E = \Theta_F \circ \varphi$.

\begin{example}[Product bundle]
Let $(X, \tau_X)$ be an involutive space.
Let $\mathbb{H} = \C \oplus \C j$ be the quaternions. We consider $\C^{2m} = \mathbb{H}^m$. Then the projection onto the first component $\underline{\C^{2m}} = X \times \C^{2m} \rightarrow X$ is a complex vector bundle.
Consider an anti-linear map $\Theta_0$ on the product space $X \times \C^{2m}$ given by $\Theta_0(x, v) = (\tau_X(x), jv)$.
Then the pair $(\underline{\C^{2m}}, \Theta_0)$ is a Quaternionic vector bundle over $(X, \tau_X)$.
\end{example}

We now assume $X$ to be a compact Hausdorff space.
The isomorphism classes of Quaternionic vector bundles\footnote{We consider only such bundles of finite rank underlying complex vector bundles.} $\Vect_Q(X, \tau_X)$ over the involutive space $(X, \tau_X)$ makes an abelian monoid by taking direct sum as its binary operation.
We define the $KSp$-group $KSp^0(X,\tau_X)$ for $(X, \tau_X)$ as the Grothendieck group of this abelian monoid.\footnote{$KSp$-groups are sometimes called $KQ$-groups \cite{NG1}. In this paper, we follow the original notation by Dupont \cite{Du69}. Note that $KSp^0(X, \tau_X)$ is naturally isomorphic to $KR^{-4}(X, \tau_X)$.}
We denote by $[(E, \Theta)]$ the class of a Quaternionic vector bundle $(E, \Theta)$ over $(X, \tau_X)$ in $KSp^0(X, \tau_X)$.
As in the case of $K$-groups, we can define
the reduced $KSp$-group $\tKSp(X, \tau_X)$ for an involutive space $((X, \tau_X), x)$ with a base point $x \in X^{\tau}$,
the $KSp$-group $KSp^0((X, \tau_X), (A,\tau_X))$ for a pair of $(X, \tau_X)$ and its closed subspace $A$ with $\tau_X(A) = A$,
and the compactly supported $KSp$-group $KSp^0_\cpt(Y, \tau_Y)$ for a locally compact Hausdorff involutive space $(Y, \tau_Y)$.
Let $n$ be a positive integer. For such a pair $(X, \tau_X)$ and $(A, \tau_X)$, we define,
\begin{equation*}
KSp^{-n}((X, \tau_X), (A, \tau_X)) := KSp^0((X, \tau_X) \times \D^{0,n}, (X, \tau_X) \times \SSS^{0,n} \cup (A, \tau_X) \times \D^{0,n}).
\end{equation*}
Then there is a long exact sequence associated to $((X, \tau_X), (A, \tau_X))$.
By the {\em Bott periodicity theorem} \cite{At66}, we have the following natural isomorphism,
\begin{equation*}
	KSp^0_\cpt(Y, \tau_Y) \cong KSp^0_\cpt((Y, \tau_Y) \times \R^{1,1}).
\end{equation*}
There also is an alternative description of $KSp$-groups. A pair of Quaternionic vector bundles $(E, \Theta_E)$ and $(F, \Theta_F)$ over $(Y,\tau_Y)$ and a ($\Z_2$-equivariant) bundle homomorphism $f \colon (E, \Theta_E) \rightarrow (F, \Theta_F)$ which is invertible outside a compact set defines an element of the $KSp$-group $KSp^0_\cpt(Y,\tau_Y)$. We denote its class by $[(E, \Theta_E),$ $(F, \Theta_F); f]$.

Let $(H, \Theta)$ be a pair of a separable complex Hilbert space $H$ and an anti-linear map $\Theta \colon H \rightarrow H$ which satisfies $\Theta \Theta^* = \id$ and $\Theta^2 = -\id$.
The map $\Theta$ induces an involution $\tau_\Theta$ on $B(H)$ given by $\tau_\Theta(G) = \Theta G \Theta^*$.
$\tau_\Theta$ is closed on $\Fred(H)$ and the pair $(\Fred(H), \tau_\Theta)$ is an involutive space.
Then we have a natural bijection $\mathrm{index} \colon [(X, \tau_X), (\Fred(H), \tau_\Theta)]_{\Z_2} \rightarrow KSp^0(X, \tau_X)$ where $[ - , - ]_{\Z_2}$ is the homotopy classes of $\Z_2$-equivariant maps \cite{Se69,Ma71}.
Let $\C^{2m} = \mathbb{H}^m$. The action of $j$ induces an involution $\tau_j$ on $\GL(2, \C)$. By taking the inductive limit of a sequence $(\GL(2, \C), \tau_j) \rightarrow (\GL(4, \C), \tau_j) \rightarrow \cdots$, we have an involutive space $(\GL(\infty, \C), \tau_j)$.
Then we have a natural isomorphism $[(X, \tau_j), (\GL(\infty, \C), \tau_j)]_{\Z_2} \cong KSp^{-1}(X, \tau_X)$.

Let $(X, \tau_X)$ and $(Y, \tau_Y)$ be involutive manifolds possibly with boundary.
We assume that its boundaries $\partial X$ and $\partial Y$ are closed under involutions $\tau_X$ and $\tau_Y$, respectively.
We further assume that there is a $\Z_2$-equivariant neat embedding $f \colon (X, \tau_X) \rightarrow (Y, \tau_Y)$ whose normal bundle $(N, \tau_N)$ is a trivial bundle of fiber $\R^{n,n}$.
We identify the normal bundle $(N, \tau_N)$ with a ($\Z_2$-equivariant) tubular neighborhood of this embedding.
\begin{definition}
We define the map $f^{\AII}_! \colon (X, \tau_X) \rightarrow (Y, \tau_Y)$
by the composition of the Bott periodicity isomorphism $KSp^0_\cpt(X, \tau_X) \rightarrow KSp^0_\cpt((X, \tau_X) \times \R^{n,n}) \cong KSp^0_\cpt(N, \tau_N)$ and the map $\mathrm{ext} \colon KSp^0_\cpt(N, \tau_N) \rightarrow KSp^0_\cpt(Y, \tau_Y)$ induced by the collapsing map $Y^+ \rightarrow Y^{+}/(Y^+ - N) \cong N^+$.
\end{definition}

\subsubsection{An explicit construction of the isomorphism $\tKSp(\SSS^{2,1}) \cong \Z_2$}
\label{sec:2.1.3}
We first note that the reduced $KSp$-group $\tKSp(\SSS^{2,1})$ is calculated as follows,
\begin{equation*}
\tKSp(\SSS^{2,1}) \cong KSp^0(\D^{2,0}, \SSS^{2,0})\cong \Z_2.
\end{equation*}
Thus there is the unique group isomorphism between $\tKSp(\SSS^{2,1})$ and $\Z_2$.
In this subsection, we give an explicit construction of this map and give an explicit expression of the non-zero element in $\tKSp(\SSS^{2,1})$.
Let $(E, \Theta)$ be a Quaternionic vector bundle over the involutive space $\SSS^{2,1}$.
The space $\SSS^{2,1}$ has two fixed point sets of its involution for which we write $0$ and $\infty$.
Since $\SSS^{2,1}$ has fixed points, the complex rank of $E$ is even, for which we write $2n$.
We consider a (complex) determinant $\det E$. Then $\Theta$ induces an anti-linear involution $\det \Theta$ on $\det E$ in a natural way, and the pair $(\det E, \det \Theta)$ is a Real vector bundle over $\SSS^{2,1}$.
Let us divide $\SSS^{2,1} = \D^{2,0}_0 \cup_{\SSS^{2,0}} \D^{2,0}_\infty$ into an upper hemisphere $\D^{2,0}_0 \ni 0$ and a lower hemisphere $\D^{2,0}_\infty \ni \infty$.
Since $\D^{2,0}$ is $\Z_2$-equivariant contractible, restrictions of $(E, \Theta)$ and $(\det E, \det \Theta)$ onto $\D^{2,0}_0$ and $\D^{2,0}_\infty$ are trivial (see \cite{NG2} \cite{NG1}, for example).
Since $Sp(n)$ is connected, the bundle $(E, \Theta)$ has trivializations on $0$ and $\infty$ which are unique up to homotopy, that is,
\begin{equation*}
	(E_0, \Theta) \cong \{ 0 \} \times \mathbb{H}^n, \quad (E_\infty, \Theta) \cong \{ \infty \} \times \mathbb{H}^n.
\end{equation*}
Above trivializations induce canonical trivializations on $(\det E_0, \det \Theta)$ and $(\det E_\infty,$ $\det \Theta)$ which are unique up to homotopy.
These trivializations extend to trivializations of $(\det E, \det \Theta)$ on $\D^{2,0}_0$ and $\D^{2,0}_\infty$, respectively.
Such trivializations are unique up to homotopy.
By using these trivializations, we obtain a clutching function $f \colon \SSS^{2,0} \rightarrow (\U(1), \tau)$ of the bundle $E$,
where $\tau(z) = \bar{z}$. It is easily checked that $f$ is a $\Z_2$-equivariant continuous map.
Thus we have a map,
$\mathrm{Vect}_Q(\SSS^{2,1}) \rightarrow [\SSS^{2,0}, (\U(1), \tau)]_{\Z_2}$.
We give a group structure on $[\SSS^{2,0}, (\U(1), \tau)]_{\Z_2}$ by using the multiplication on $\U(1)$.
Since $\det(f \oplus g) = \det f \cdot \det g$, the above map is a monoid homomorphism.
Let $f \colon \SSS^{2,0} \rightarrow (\U(1), \tau)$ be a $\Z_2$-equivariant map.
Then there exists some $t \in \SSS^{2,0}$ such that $f(t) \in \{ \pm 1 \}$.
We choose a path $l$ in $\SSS^{2,0}$ which connects $t$ and $-t$, and fix an orientation on $l$ and $\U(1)$.
Since $f(-t) = \overline{f(t)} = f(t)$, the map $f|_l \colon l \rightarrow \U(1)$ defines a loop in $\U(1)$, and we consider the winding number for this loop.
It is easily checked that the mod $2$ of this winding number (for which we write $w(f) \in \Z_2$) does not depend on the choice made.
Thus we have a map
$w \colon [\SSS^{2,0}, (\U(1), \tau)]_{\Z_2} \rightarrow \Z_2$.
It is easily checked that the map $w$ is a group homomorphism.
By the composition of the above maps, we obtain a monoid homomorphism $\mathrm{Vect}_Q(\SSS^{2,1}) \rightarrow \Z_2$.
Thus by the universality of the Grothendieck construction, we obtain a group homomorphism,
$KSp^0(\SSS^{2,1}) \rightarrow \Z_2$.
It is clear by construction that this map maps classes of product bundles to zero.
Thus we have a group homomorphism $B \colon \tKSp(\SSS^{2,1}) \rightarrow \Z_2$.

\begin{example}\label{localizationexample}
On $\C^2$, we consider a left $\mathbb{H}$-module structure given by,
\begin{equation*}
	I = \left( 
    \begin{array}{cc}
           i&0\\
           0&i
    \end{array}
\right),
\quad
	J = \left( 
    \begin{array}{cc}
           0 & -\tau\\
           \tau & 0
    \end{array}
    \right),
\end{equation*}
where $\tau$ is the complex conjugation.
The involution $\tau_J$ on $M(2, \C)$ is given by
\begin{equation*}
	\tau_J \left( 
    \begin{array}{cc}
           x&y\\
           z&w
    \end{array}
\right)
	=
		J \left( 
    \begin{array}{cc}
           x&y\\
           z&w
    \end{array}
\right) J^*
		=
			\left( 
    \begin{array}{cc}
           \bar{w}&-\bar{z}\\
           -\bar{y}&\bar{x}
    \end{array}
\right).
\end{equation*}

We now identify $[0, 2\pi]/\{ 0, 2\pi \}$ with $\SSS^{2,0}$ by $\theta \mapsto e^{i \theta}$. Then the involution on $\SSS^{2,0}$ is translated into the involution on $[0, 2\pi]/\{ 0, 2\pi \}$ given by $\theta \mapsto \theta \pm \pi$.
Let $k$ be an integer.
We define a map $g_k \colon [0, 2\pi]/\{ 0, 2\pi \} \rightarrow (M(2, \C), \tau_J)$ by,
\begin{equation*}
	g_k(\theta) =
\left\{
\begin{aligned}
\mathrm{diag}(e^{2ik\theta}, 1)
 & \qquad \text{if} \quad 0 \leq \theta \leq \pi,   \\
 \mathrm{diag}(1, e^{-2ik\theta})
 & \qquad \text{if} \quad \pi \leq \theta \leq 2\pi.
\end{aligned}
\right.
\end{equation*}
The map $g_k$ is $\Z_2$-equivariant
and the triple $[(\underline{\C^2}, J), (\underline{\C^2}, J) ; g_k]$ defines an element in $KSp^0(\D^{2,0}, \SSS^{2,0})$.
The determinant of $g_k$ is,
\begin{equation*}
	\det g_k (\theta) =
\left\{
\begin{aligned}
e^{2ik\theta} & \qquad \text{if} \quad 0 \leq \theta \leq \pi,  \\
e^{-2ik\theta} & \qquad \text{if} \quad \pi \leq \theta \leq 2\pi,
\end{aligned}
\right.
\end{equation*}
which is $\Z_2$-equivariant.
It is easily checked that $w(\det g_k) = k \mod 2 \in \Z_2$.
\end{example}

By this example, the following holds.
\begin{proposition}\label{isom}
	The map $B \colon \tKSp(\SSS^{2,1}) \rightarrow \Z_2$ is a group isomorphism.
\end{proposition}

\subsection{Spectral flow}
\label{sec:2.2}
In this subsection, we give a definition of the $\Z$-valued spectral flow and the $\Z_2$-valued spectral flow .
They are used to define edge indices for type A and type $\AII$ topological insulators, respectively.
\subsubsection{$\Z$-valued spectral flow}
The $\Z$-valued spectral flow is, roughly speaking, the net number of crossing points of eigenvalues of the family of self-adjoint Fredholm operators with zero.\footnote{Counted positively for increasing eigenvalue crossings with respect to the parameter.}
\begin{definition}[\cite{Ph96}]\label{sf}
Let $\Fred(H)^{s.a.}$ be the space of self-adjoint Fredholm operators on a fixed separable complex Hilbert space $H$ with norm topology.
Let $B \colon [-\pi,\pi] \rightarrow \Fred(H)^{s.a.}$ be a continuous path.
We choose a partition $-\pi = t_0 < t_1 < \cdots < t_n = \pi$ and positive numbers $c_1, c_2, \ldots, c_n$ so that for each $i = 1, 2, \ldots, n$, the function $t \mapsto \chi_{[-c_i, c_i]}(B_t)$ is continuous and finite rank on $[t_{i-1}, t_i]$, where $\chi_{[a, b]}$ denotes the characteristic function of $[a, b]$.
We define the spectral flow of $B$ by
\begin{equation*}
	\mathrm{sf}(B) := \sum^{n}_{i=1} ( \rank_\C(\chi_{[0, c_i]}(B_{t_i})) - \rank_\C(\chi_{[0, c_i]}(B_{t_{i-1}})) \in \Z.
\end{equation*}
\end{definition}
The spectral flow is independent of the choice made and depends only on the homotopy class of the path $B$ leaving the endpoints fixed.
\subsubsection{$\Z_2$-valued spectral flow}
Let $(H, \Theta)$ be a pair of a separable complex Hilbert space $H$ and an anti-linear operator $\Theta \colon H \rightarrow H$ which satisfies $\Theta \Theta^* = \id_H$ and $\Theta^2 = -\id_H$.
Note that the space of self-adjoint Fredholm operators $\Fred(H)^{s.a.}$ on $H$ is closed under the involution $\tau_\Theta$.
We regard the interval $[-\pi,\pi]$ as an involutive space by the involution $\tau_I(t) = -t$. We define, for a $\Z_2$-equivariant continuous map $B \colon ([-\pi,\pi], \tau_I) \rightarrow (\Fred(H)^{s.a.}, \tau_\Theta)$, a $\Z_2$-valued spectral flow.
\begin{definition}\label{qsf}
Let $B \colon ([-\pi,\pi], \tau_I) \rightarrow (\Fred(H)^{s.a.}, \tau_\Theta)$ be a $\Z_2$-equivariant continuous map.
We choose a partition $0 = t_0 < t_1 < \cdots < t_n = \pi$ of $[0,\pi]$ and positive numbers $c_1, c_2, \ldots, c_n$ so that for each $i = 1, 2, \ldots, n$, the function $t \mapsto \chi_{[-c_i, c_i]}(B_t)$ is continuous and finite rank on $[t_{i-1}, t_i]$.
We then define the {\em $\Z_2$-valued spectral flow} $\qsf(B)$ of $B$ as follows.
\begin{equation*}
	\qsf(B) := \sum^{n}_{i=1} ( \rank_\C(\chi_{[0, c_i]}(B_{t_i})) + \rank_\C(\chi_{[0, c_i]}(B_{t_{i-1}})))
	\hspace{2mm} \bmod{2} \quad \in \Z_2.
\end{equation*}
\end{definition}
The well-definedness and the homotopy invariance\footnote{$\Z_2$-valued spectral flow is invariant for a $\Z_2$-equivariantly homotopic path leaving the endpoints fixed or leaving these points in $(\Fred(H)^{s.a.})^{\tau_\Theta}$. This results from Kramers degeneracy.}
 of $\qsf(B)$ are also be proved easily in the same way as \cite{Ph96}.
Note that De Nittis--Schulz-Baldes considered in \cite{NSB15} such $\Z_2$-valued spectral flow for more general class of operators and symmetries.
$\Z_2$-valued spectral flow presented here is different from the one considered in \cite{CPS16}.
Our concrete description here is just for the later use.
For such a path $B \colon ([-\pi,\pi], \tau_I) \rightarrow (\Fred(H)^{s.a.}, \tau_\Theta)$, its spectrum is symmetric with respect to the involution $\tau_I$, that is, $\mathrm{sp}(B_t) = \mathrm{sp}(B_{-t})$.
By this symmetry, crossing points of eigenvalues with zero appear as pairs, and $\qsf(B)$ is mod $2$ of the net number of such pairs with multiplicity.

\begin{example}
Let $B$ be a family of linear operators on $\C^2$ given by $B_t = \left(
\begin{array}{cc}
           t&0\\
           0& -t
\end{array}
\right)$ for $t \in [-\pi, \pi]$.
Let $\Theta = \left(
\begin{array}{cc}
           0&-\tau\\
           \tau& 0
\end{array}
\right)$, where $\tau$ is the complex conjugation on $\C$.
Then its $\Z_2$-valued spectral flow is defined and it can be easily checked that $\qsf(B) = 1$.
\end{example}

\section{Setup}
\label{sec:3}
Let $\T$ be the unit circle in the complex plane.
We write $t$ for an element of $\T$. We fix the counter-clockwise orientation on the circle $\T$. This oriented circle has a unique spin$^c$ structure up to isomorphism.
Let $V$ be a finite dimensional Hermitian vector space. We denote by $N$ the complex dimension of $V$, and $|| \cdot ||_V$ for its norm.
Let $l^2(\Z; V)$ be the space of sequences $\varphi = \{ \varphi_n \}_{n \in \Z}$ where each $\varphi_n$ is an element of $V$, and satisfies $\sum_{n \in \Z} ||\varphi_n||_V^2 < + \infty$.
Let $S \colon l^2(\Z; V) \rightarrow l^2(\Z; V)$ be the shift operator defined by $(S\varphi)_n = \varphi_{n-1}$.
Let $A_j \colon \T \rightarrow \End_\C(V)$ $(j \in \Z)$ be continuous maps which satisfies $\sum_{j \in \Z}||A_j||_{\infty} < +\infty$, where $||A_j||_{\infty} = \sup_{t \in \T}||A_j(t)||_V$.
We consider a tight-binding Hamiltonian on the lattice $\Z \times \Z$ with the periodic potential.
\begin{definition}[Bulk Hamiltonian]\label{def3.1}
For each $t$ in $\T$, we define a linear map,
\begin{equation*}
	H(t) \colon l^2(\Z; V) \rightarrow l^2(\Z; V),
\end{equation*}
by $(H(t)\varphi)_n = \sum_{j \in \Z}A_j(t) \varphi_{n-j}$. We call $H(t)$ a {\em bulk Hamiltonian}. $H(t)$ is a bounded linear operator and commutes with the shift operator $S$.
\end{definition}

We assume that $H(t)$ is a self-adjoint operator for any $t$ in $\T$. Thus its spectrum $\sigma(H(t))$ is contained in $\R$. The following is our central assumption.

\begin{assumption}[Spectral gap condition]
We assume further that our bulk Hamiltonian has a spectral gap at the Fermi level $\mu$ for any $t$ in $\T$, i.e. $\mu$ is not contained in $\sigma(H(t))$. We call this condition as a {\em spectral gap condition}.
\end{assumption}

We take a positively oriented simple closed smooth loop $\gamma$ through $\mu$ in the complex plane, which satisfies that, for any $t$ in $\T$, the set $\sigma(H(t)) \cap (-\infty, \mu)$ is contained inside the loop $\gamma$, and $\sigma(H(t))\cap (\mu, +\infty)$ is contained outside $\gamma$ (see Figure $1$).
This oriented loop has a unique spin$^c$ structure up to isomorphism.
Let $\SSS^1_\eta$ be the unit circle in the complex plane, and $H(\eta, t) \in \End_\C(V)$ be a continuous family of Hermitian operators on $\SSS^1_\eta \times \T$ given by $H(\eta, t) = \sum_{j \in \Z}A_j(t)\eta^j$.
Since the bulk Hamiltonian commutes with the shift operator $S$, by the Fourier transform $l^2(\Z; V) \cong L^2(\SSS^1_\eta; V)$, $H(t)$ is expressed as a multiplication operator $M_{H(\eta, t)}$ on $L^2(\SSS^1_\eta; V)$ generated by $H(\eta, t)$.
Note that the shift operator $S$ corresponds to $M_\eta$ by the Fourier transform.\footnote{Note that, Graf and Porta used in \cite{GP13}, the operator $\tilde{S}$ defined by $(\tilde{S}\varphi)_n = \varphi_{n+1}$ instead of our shift operator $S$. Consequently, our choice of the orientation of $\SSS^1_\eta$ (in Sect.~\ref{sec:4}) is different from \cite{GP13}. Although our choice of $S$ causes such confusing differences with \cite{GP13}, we decided to use this $S$ for the shift operator (and the orientation on $\SSS^1_\eta$), since this choice is consistent with the one used in the field of Toeplitz operators \cite{BS06}, and we use a family index  of a family of Toeplitz operators in the proof of Lemma \ref{lem5}.}
The following lemma follows from the spectral gap condition.
\begin{lemma}\label{lem1}
For $(\eta, z, t) \in \SSS^1_{\eta} \times \gamma \times \T$, $H(\eta, t) - z$ is invertible.
\end{lemma}
\begin{figure}[htbp]
  \begin{center}
    \includegraphics[width=84mm,clip]{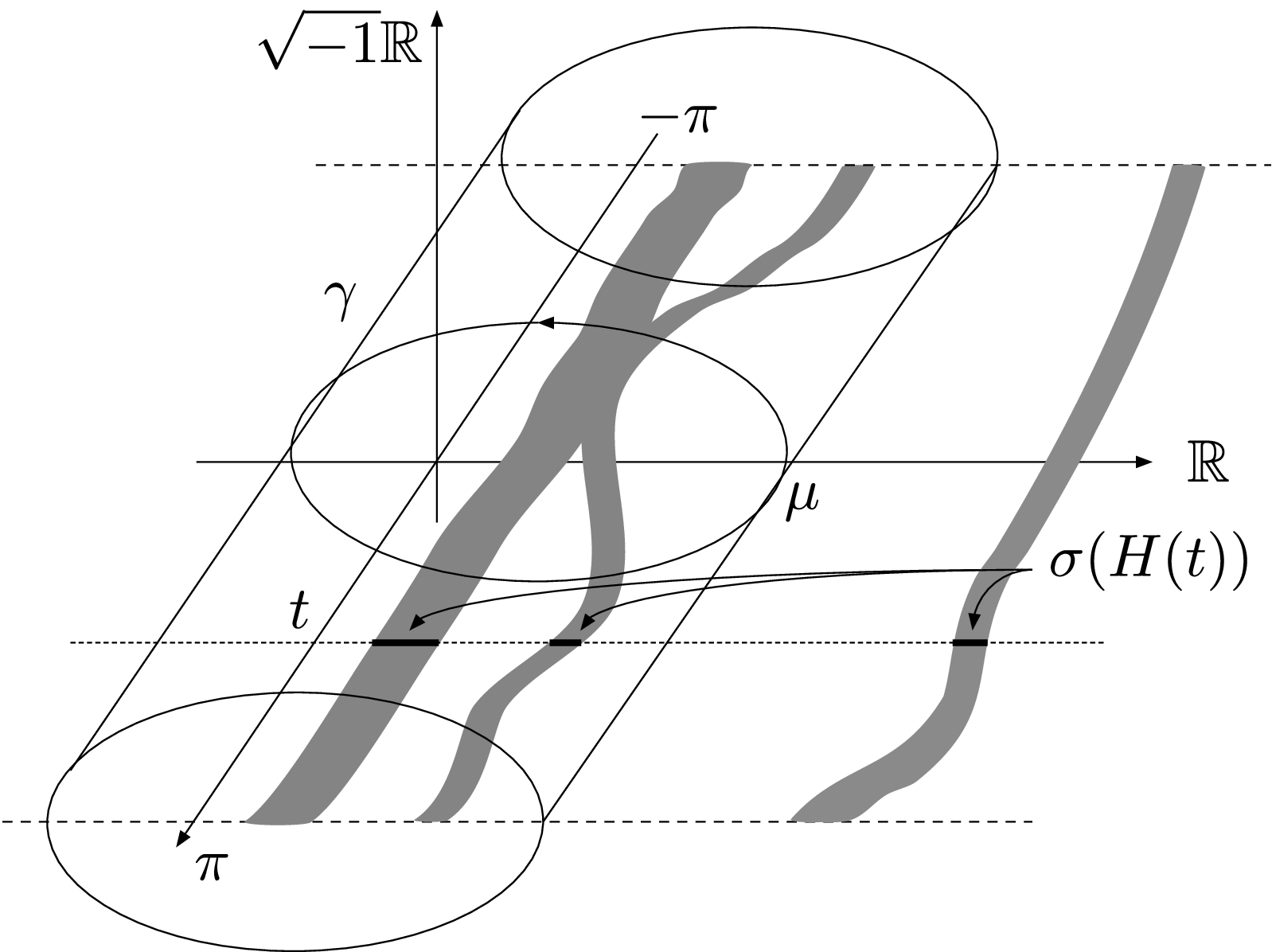}
	\caption{For an element $t$ in $\T$, bold black lines indicate the spectrum of the bulk Hamiltonian $\sigma(H(t))$. The family of spectra of bulk Hamiltonians makes a gray area. In order to make this picture simple, we here consider a family of spectra parametrized by an interval $[-\pi, \pi]$ by considering the pullback by the map $[-\pi, \pi] \rightarrow [-\pi, \pi]/\{ -\pi, \pi \} \cong \T$. Therefore dotted lines at ${\pm \pi}$ should be glued. The loop $\gamma$ (and the torus $\gamma \times \T$) is indicated as a (family of) circle(s)}
  \end{center}
\end{figure}
\begin{remark}\label{rem2}
As we will see later, taking this loop $\gamma$ is substantial for the definition of the bulk index and the edge index. In this sense, we do not need to choose this loop in order to formulate the bulk-edge correspondence for our system.
We need to take this loop in order to define a vector bundle consisting of decaying solutions.
\end{remark}

\section{Bulk index, edge index and the bulk-edge correspondence}
\label{sec:4}
We define the bulk index and the edge index, and formulate our main theorem.
\begin{definition}[Bloch bundle]\label{blochbundle}
By the spectral gap condition, Riesz projections give a continuous family of projections on $V$ parametrized by $\SSS^1_\eta \times \T$, given by $\frac{1}{2\pi i} \int_\gamma(\lambda \cdot 1 - H(\eta, t))^{-1}d\lambda$.
The images of this family makes a complex vector subbundle $E_\B$ of the product bundle $\underline{V} = (\SSS^1_\eta \times \T) \times V$. $E_\B$ is called the {\em Bloch bundle}.
\end{definition}
The Bloch bundle defines an element $[E_\B]$ of the $K$-group $K^0(\SSS^1_\eta \times \T)$.
We fix a counter-clockwise orientation on $\SSS^1_\eta$. This oriented circle has a unique spin$^c$ structure up to isomorphism.
We consider the product spin$^c$ structure on $\SSS^1_\eta \times \T$.
By using this spin$^c$ structure, we have the map
$\ind_{\SSS^1_\eta \times \T} \colon K^0(\SSS^1_\eta \times \T) \rightarrow \Z$.
\begin{definition}[Bulk index]\label{bulkindex}
We define the {\em bulk index} of our system by
\begin{equation*}
	\I_{\Bulk} := - \ind_{\SSS^1_\eta \times \T}([E_\B]).
\end{equation*}
\end{definition}
\begin{remark}\label{rem3}
Note that the determinant line bundle associated to this spin$^c$ structure on $\SSS^1_\eta \times \T$ is trivial.
By the Atiyah-Singer index formula \cite{AS3}, our bulk index coincides with the minus of the first Chern number of the Bloch bundle. This minus sign is caused by our choice of orientation of $\SSS^1_\eta$. Note that our bulk index equals to the one considered by Graf--Porta \cite{GP13} and also the TKNN number.
\end{remark}

For each $k \in \Z$, let $\Z_{\geq k} := \{ k, k+1, k+2, \ldots \}$ be the set of all integers greater than or equal to $k$.
$l^2(\Z_{\geq k}; V)$ is a closed subspace of $l^2(\Z; V)$ in a natural way.
For each $k \in \Z$, let $P_{\geq k}$ be the orthogonal projection of $l^2(\Z; V)$ onto $l^2(\Z_{\geq k}; V)$.
\begin{definition}[Edge Hamiltonian]
For each $t$ in $\T$, we consider an operator $H^\#(t)$ given by the compression of  $H(t)$ onto $l^2(\Z_{\geq 0}; V)$, that is,
\begin{equation*}
	H^\#(t) := P_{\geq 0} H(t) P_{\geq 0} \colon l^2(\Z_{\geq 0} ; V) \rightarrow l^2(\Z_{\geq 0} ; V).
\end{equation*}
We call $H^\#(t)$ an {\em edge Hamiltonian}.
\end{definition}

$H^\#(t)$ is the Toeplitz (or the discrete Wiener-Hopf) operator with continuous symbol $H(\eta, t)$.
Since $H(\eta, t)$ is Hermitian, $\{ H^\#(t) - \mu \}_{t \in \T}$ is a norm-continuous family of self-adjoint Fredholm operators on the Hilbert space $l^2(\Z_{\geq 0}; V)$.
\begin{definition}[Edge index]\label{edgeindex}
We define the {\em edge index} of our system as the minus of the spectral flow of the family $\{ H^\#(t) - \mu \}_{t \in \T}$, that is,
\begin{equation*}
	\I_{\Edge} := -\mathrm{sf}(\{ H^\#(t) - \mu \}_{t \in \T}).
\end{equation*}
\end{definition}

We now revisit the definition of the spectral flow (Definition \ref{sf}) and the edge index (see Figure $2$). In order to define $\mathrm{sf}(\{ H^\#(t) - \mu \}_{t \in \T})$, we need to choose $t_i$ and $c_i$.
These data are used to prove our main theorem, and we need to choose such data in the following specific way.
We regard $\{ H^\#(t) - \mu \}_{t \in \T}$ as a family parametrized by $[-\pi, \pi]$, and
consider the following path in $\C \times [-\pi, \pi]$ (or a loop in $\C \times \T$),
\begin{equation*}
	l := \bigcup_{i=0}^{n-1} \mathcal{L}(\tilde{c}_i, \tilde{c}_{i+1}) \times \{ t_i\} \cup \{ \tilde{c}_{i+1} \} \times [t_i, t_{i+1}] \bigcup \mathcal{L}(\tilde{c}_n, \tilde{c}_0) \times \{ t_n \},
\end{equation*}
where $c_0 := 0$, and we set $\tilde{c}_i := c_i +\mu$, for simplicity.
We denote by $\mathcal{L}(a, b)$ the closed interval in $\R$ between $a$ and $b$.
The loop $l$ may have intersections with spectra of edge Hamiltonians.
If we choose $t_i$ and $c_i$ as in Definition \ref{sf}, crossing points appear only at intervals of the form $\mathcal{L}(\tilde{c}_i, \tilde{c}_{i+1}) \times \{ t_i\}$.
We take $t_i$ and $c_i$ so that crossing points appear only at open intervals of the form $\interior{\mathcal{L}}(\tilde{c}_i, \tilde{c}_{i+1}) \times \{ t_i\}$ ($i = 1, \ldots, n-1$), where $\interior{\mathcal{L}}(a, b)$ denotes the open interval in $\R$ between $a$ and $b$,
or at half-open intervals $[\tilde{c}_0, \tilde{c}_1) \times \{ t_0\}$ and $[\tilde{c}_0, \tilde{c}_n ) \times \{ t_n \}$.\footnote{The existence of $t_i$ and $c_i$ which satisfies these conditions follows by the argument in \cite{Ph96}.}
We see that, when we consider $\mathrm{sf}(\{ H^\#(t) - \mu \}_{t \in \T})$, we do not need to consider about crossing points at $\{ \tilde{c}_0 \} \times \{ t_0 \}$ and $\{ \tilde{c}_0 \} \times \{ t_n \}$.
This is because crossing points at these points, if exist, do not contribute to the spectral flow (actually, they cancel out, in our case).
We assume that there are no crossing points at these points, for simplicity.
Then, $\mathrm{sf}(\{ H^\#(t) - \mu \}_{t \in \T})$ is the net number of crossing points $\{ u_1, \ldots, u_m \}$ counted with multiplicity, where crossing points in $\mathcal{L}(\tilde{c}_i, \tilde{c}_{i+1}) \times \{ t_i\}$  with $\tilde{c}_i > \tilde{c}_{i+1}$ are counted with positive sign, and that of $\tilde{c}_i < \tilde{c}_{i+1}$ are counted with negative sign.
We fix an orientation on $l$ which is compatible with the natural one on the interval $\{ \tilde{c}_{i+1} \} \times [t_i, t_{i+1}]$ (see Figure $2$).
\begin{remark}\label{rem4}
At each crossing point, there exist {\em edge states} which mean eigenvectors of our edge Hamiltonian of eigenvalue $\mu$. In this sense, the edge index is defined by counting edge states.
\end{remark}

\begin{figure}[htbp]
  \begin{center}
    \includegraphics[width=84mm,clip]{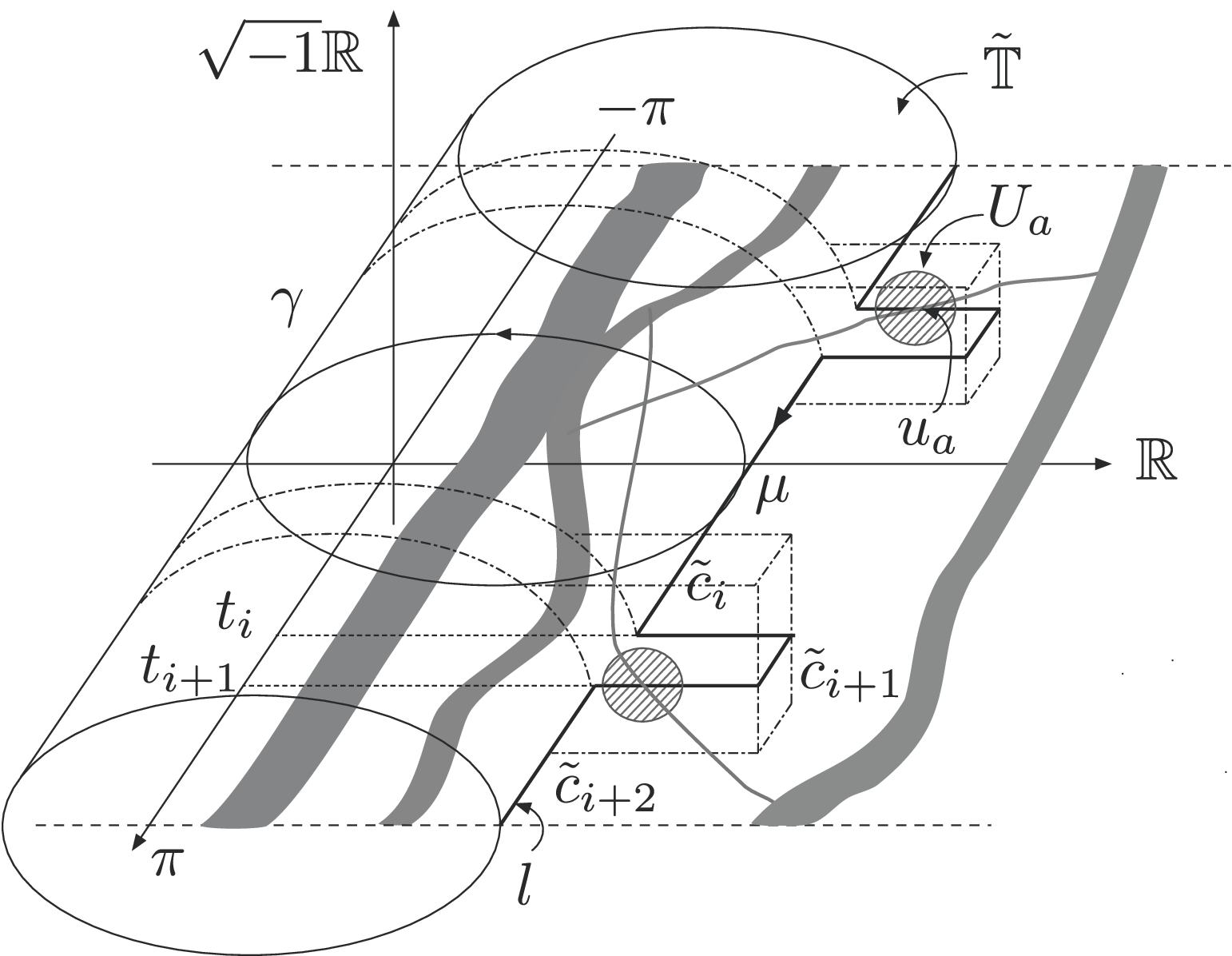}
	\caption{The spectrum of edge Hamiltonians (gray area and gray line), the loop $l$ and the deformed torus $\tilde{\T}$.
The edge index is calculated by counting crossing points $\{ u_1, \ldots, u_m \}$  of the gray line and the loop $l$ taking multiplicity and sign into account.
The deformed torus intersects with the component of $\R \times \T$ in the gap across the loop $l$
}
  \end{center}
\end{figure}

Our main theorem, the bulk-edge correspondence for two-dimensional type A topological insulators is the following:
\begin{theorem}[Bulk-edge correspondence for quantum Hall systems]\label{mainthm}
The bulk index coincides with the edge index. That is,
$\I_{\Bulk} = \I_{\Edge}$.
\end{theorem}
To give a proof of this theorem is the purpose of the rest of this paper.
\section{A vector bundle consisting of decaying solutions}
\label{sec:5}
In this section, we define a vector bundle over $\gamma \times \T$, which is a generalization of the one considered by Graf and Porta in \cite{GP13}.
\begin{lemma}\label{lem2}
	There exists a positive integer $K$ such that for any integer $k \geq K$ and $(z, t)$ in $\gamma \times \T$, the following map is surjective.
\begin{equation*}
	P_{\geq k}(H(t) - z)P_{\geq 0} \colon l^2(\Z_{\geq 0} ; V) \rightarrow l^2(\Z_{\geq k} ; V).
\end{equation*}
\end{lemma}
\begin{proof}
For a non-negative integer $k$, we define a linear map
$H_k(t) \colon l^2(\Z; V) \rightarrow l^2(\Z; V)$
by
$(H_k(t)\varphi)_n = \sum_{j=-k}^{j=k}A_j(t) \varphi_{n-j}$.
We first show that there exists some integer $K'$ such that for any $k \geq K'$ and for any $(z, t) \in \gamma \times \T$, the operator $P_{\geq k}(H_k(t) - z)P_{\geq 0} \colon l^2(\Z_{\geq 0} ; V) \rightarrow l^2(\Z_{\geq k} ; V)$ is surjective.
Note that $|| (H(t)-z) - (H_k(t)-z)|| \leq \sum_{|j| > k}||A_j||_{\infty}  \rightarrow 0$ as $k \rightarrow \infty$.
Since $H(t)-z$ is invertible by Lemma~\ref{lem1}, there exists a positive integer $K'$ such that for any $k \geq K'$ and $(z, t)$ in $\gamma \times \T$,  $H_k(t) - z$ is invertible.
Note that $S P_{\geq k} = P_{\geq k+1} S$, and we have $P_{\geq 0} = (S^*)^k P_{\geq k} S^k$.
Since $(H_k(t) - z) \eta^{-k}$ does not contain the term of $\eta^m$ $(m > 0)$, $(H_k(t) - z)^{-1}\eta^k$ does not contain the term of $\eta^m$ $(m < 0)$, and so $(H_k(t) - z)^{-1}M_{\eta^{k}}$ maps $l^2(\Z_{\geq k}; V)$ to $l^2(\Z_{\geq k}; V)$.
Thus for $\varphi \in l^2(\Z_{\geq k}; V)$,
\begin{gather*}
	(P_{\geq k} (H_k(t) - z) P_{\geq 0})(P_{\geq 0 }(H_k(t) - z)^{-1}P_{\geq k})(\varphi) \hspace{20mm} \\
		\hspace{10mm} = P_{\geq k} (H_k(t) - z) (S^*)^k P_{\geq k} S^k (H_k(t) - z)^{-1}(\varphi)\\
		\hspace{12.5mm} = P_{\geq k} (H_k(t) - z) M_{\eta^{-k}} P_{\geq k} M_{\eta^{k}} (H_k(t) - z)^{-1}(\varphi)\\
		\hspace{12.5mm} = P_{\geq k} (H_k(t) - z) M_{\eta^{-k}} P_{\geq k} (H_k(t) - z)^{-1}M_{\eta^{k}}(\varphi)\\
		\hspace{14.5mm} = P_{\geq k} (H_k(t) - z) M_{\eta^{-k}} (H_k(t) - z)^{-1}M_{\eta^{k}}(\varphi) = \varphi.
\end{gather*}
Thus $P_{\geq k} (H_k(t) - z) P_{\geq 0}$ has a right inverse, and it is surjective.\footnote{Such argument was used by Atiyah and R. Bott to compute the family index of a family of Toeplitz operators \cite{At68,AB64}.}

As is shown above, the composite of $P_{\geq 0 }(H_k(t) - z)^{-1}P_{\geq k}$ and $P_{\geq k} (H_k(t) - z) P_{\geq 0}$ is the identity map. Thus $P_{\geq 0 }(H_k(t) - z)^{-1}P_{\geq k}$ is bounded below, and so closed range.
We also have
\begin{equation*}
	\mathrm{Im}(P_{\geq 0 }(H_k(t) - z)^{-1}P_{\geq k}) \cap \Ker(P_{\geq k} (H_k(t) - z) P_{\geq 0}) = \{ 0 \}, 
\end{equation*}
\begin{equation*}
	\mathrm{Im}(P_{\geq 0 }(H_k(t) - z)^{-1}P_{\geq k}) \oplus \Ker(P_{\geq k} (H_k(t) - z) P_{\geq 0}) = l^2(\Z_{\geq 0}; V).
\end{equation*}
Let $X_k := \mathrm{Im}(P_{\geq 0 }(H_k(t) - z)^{-1}P_{\geq k})$. We denote the orthogonal projection onto $X_k$ by $P_{X_k}$.
Then by the open mapping theorem, the map $P_{\geq k} (H_k(t) - z) P_{X_k} \colon X_k \rightarrow l^2(\Z_{\geq k}; V)$ is invertible.
On the other hand, we have $P_{\geq 0 }(H_k(t) - z)^{-1}P_{\geq k} = P_{X_k}(H_k(t) - z)^{-1}P_{\geq k}$ and
\begin{eqnarray*}
	\bm{1} &=& (P_{\geq k} (H_k(t) - z) P_{\geq 0})(P_{\geq 0 }(H_k(t) - z)^{-1}P_{\geq k})\\
		&=& (P_{\geq k} (H_k(t) - z) P_{X_k})(P_{X_k}(H_k(t) - z)^{-1}P_{\geq k}).
\end{eqnarray*}
Therefore, the inverse of $P_{\geq k} (H_k(t) - z) P_{X_k}$ is $P_{X_k}(H_k(t) - z)^{-1}P_{\geq k}$.

Since $H_k(t) - z$ converges to $H(t) - z$ as $k \rightarrow \infty$ uniformly with respect to $(z, t) \in \gamma \times \T$,
the operator $(H_k(t) - z)$ converges to $H(t) - z$ uniformly, and so
$(H_k(t) - z)^{-1}$ converges to $(H(t) - z)^{-1}$ uniformly with respect to $(z, t)$.
Thus there exist some $L>0$ and an integer $K''$ such that for any $k \geq K''$ and $(z, t) \in \gamma \times \T$, we have $||(H_k(t) - z)^{-1}|| < L$.
For $k \geq K''$, we have
$||(P_{\geq k} (H_k(t) - z) P_{X_k})^{-1}||^{-1} =  ||P_{\geq k} (H_k(t) - z)^{-1} P_{X_k}||^{-1} > ||(H_k(t) - z)^{-1}||^{-1} > L^{-1} > 0$.
On the other hand, we have,
\begin{equation*}
||P_{\geq k}(H(t) - z)P_{X_k} - P_{\geq k}(H_k(t) - z)P_{X_k}|| \leq || H(t) - H_k(t)|| \leq \sum_{|j| > k}||A_j||_{\infty} \rightarrow 0.
\end{equation*}
Thus there is an positive integer $K$ such that for any $k \geq K$ and $(z, t) \in \gamma \times \T$,
\begin{equation*}
||P_{\geq k}(H(t) - z)P_{X_k} - P_{\geq k}(H_k(t) - z)P_{X_k}|| \leq || (P_{\geq k}(H_k(t) - z)P_{X_k})^{-1} ||^{-1},
\end{equation*}
which means that $P_{\geq k}(H(t) - z)P_{X_k}$ is invertible. Thus $P_{\geq k}(H(t) - z)P_{\geq 0}$ is surjective for such $k \geq K$ and $(z, t) \in \gamma \times \T$.
\end{proof}
We choose such an integer $k \geq K$, and set
$(E_{\GP})_{z, t} := \Ker (P_{\geq k}(H(t) - z)P_{\geq 0})$.
\begin{remark}\label{rem6}
For $n \geq k$, we have $((P_{\geq k}(H(t) - z)P_{\geq 0})(\varphi))_n = ((H(t) - z)(\varphi))_n$.
Thus $(E_{\GP})_{z, t}$ consists of sequences $\varphi = \{ \varphi_n \}_{n \geq 0}$ in $l^2(\Z_{\geq 0} ; V)$ which satisfies $(H(t)(\varphi))_n = z \varphi_n$ for $n \geq k$.
\end{remark}
We define an operator $H^{\flat}(z,t) \colon l^2(\Z_{\geq 0} ; V) \rightarrow l^2(\Z_{\geq 0} ; V)$ to be the composite of $P_{\geq k}(H(t) - z)P_{\geq 0}$ and the inclusion $l^2(\Z_{\geq k}; V) \hookrightarrow l^2(\Z_{\geq 0}; V)$.
\begin{lemma}\label{lem3}
For any $(z, t)$ in $\gamma \times \T$, $H^{\flat}(z, t)$ is a Fredholm operator whose Fredholm index is zero. Moreover, $\Ker H^{\flat}(z,t)$ and $(E_{\GP})_{z,t}$ are the same, and $\Coker H^{\flat}(z,t)$ is naturally isomorphic to $V^{\oplus k}$.
Thus the dimension of $(E_{\GP})_{z,t}$ is $kN$ which is constant with respect to the parameter.
\end{lemma}
\begin{proof}
For $(z, t) \in \gamma \times \T$, $(H^\#(t) - z) - H^{\flat}(z,t)$ is a finite rank operator, and $H^\#(t) - z$ is a Fredholm operator.
Thus, $H^{\flat}(z, t)$ also is a Fredholm operator, and their Fredholm indices are the same.
Since $H^\#(t) - z$ is connected to a self-adjoint Fredholm operator $H^\#(t) - \mu$ by a continuous path of Fredholm operators, the homotopy invariance of the Fredholm index says that the Fredholm index of  $H^\#(t) - z$ is zero. Thus the Fredholm index of $H^{\flat}(z, t)$ is zero, and so the dimensions of its kernel and cokernel are the same.
On the other hand, $\Ker H^{\flat}(z, t)$ coincides with $(E_{\GP})_{z,t}$ by definition.
By Lemma~\ref{lem2}, $\Coker H^{\flat}(z,t)$ is naturally isomorphic to the closed subspace of $l^2(\Z_{\geq 0}; V)$ which consists of the sequences with $\varphi_n = 0$ for $n \geq k$, which is naturally isomorphic to $V^{\oplus k}$.
\end{proof}
\begin{definition}\label{GPbundle}
By Lemma~\ref{lem3}, we have a vector bundle over $\gamma \times \T$ whose fiber at a point $(z, t)$ is given by $(E_{\GP})_{z, t}$. We denote this vector bundle by $E_{\GP}$. The rank of this vector bundle is $kN$. 
\end{definition}
Note that the family index of a family of Fredholm operators $\{ H^\flat (z,t) \}_{\gamma \times \T}$ is an element of the $K$-group $K^0(\gamma \times \T)$, and is given by $[E_{\GP}] - [\underline{V}^{\oplus k}]$.
Thus the bundle $E_{\GP}$ appears as a family index.
As in the definition of the bulk index, we consider the product spin$^c$ structure on $\gamma \times \T$.
By using this spin$^c$ structure, we have the map
$\ind_{\gamma \times \T} \colon K^0(\gamma \times \T) \rightarrow \Z$.
\begin{definition}[Graf--Porta \cite{GP13}]\label{GPindex}
Following Graf--Porta's idea, we here introduce another invariant $\I_{\GP}$ defined by,
\begin{equation*}
	\I_{\GP} := \ind_{\gamma \times \T}([E_{\GP}] - [\underline{V}^{\oplus k}]).
\end{equation*}
\end{definition}

\begin{remark}\label{rem7}
The conclusion of Lemma~\ref{lem3} also holds for $(z, t) \in \C \times \T$ which satisfies $z \notin \sigma(H(t))$.
Thus we can define a vector bundle in the same way on the space $\{ (z, t) \in \C \times \T \ | \ z \notin \sigma(H(t)) \}$. We denote this vector bundle by $\tilde{E}_{\GP}$. Then we have $E_{\GP} = \tilde{E}_{\GP}|_{\gamma \times \T}$.
Such extension is used in Sect. \ref{subsec:6.2}.
\end{remark}
\begin{remark}\label{rem8}
The determinant line bundle associated to the spin$^c$ structure on $\gamma \times \T$ is trivial.
Thus, as in Remark~\ref{rem3}, we can show that $\ind_{\gamma \times \T}([\underline{V}^{\oplus k}]) = 0$, and
$\I_\GP$ coincides with the first Chern number of the bundle $E_{\GP}$.
\end{remark}
\begin{remark}\label{rem9}
If, for some $k$, we have $A_j = 0$ for $j < -k$ and $k < j$, and if $A_{-k}$ and $A_{k}$ are invertible, then the difference equation $H(t)\varphi = z \varphi$ is solved by an initial condition. In this case, each fiber $(E_{\GP})_{z, t}$ is identified with the space of formal solutions of the equation $H(t)\psi = z \psi$ which decays as $n \to +\infty$. 
This is Graf--Porta's original idea to define such a vector bundle \cite{GP13}.
\end{remark}

\section{Proof of the bulk-edge correspondence}
\label{sec:6}
In this section, we prove Theorem~\ref{mainthm} by first showing that the bulk index coincides with the minus of $\I_{\GP}$ (Sect.~\ref{subsec:6.1}) and next showing that the minus of $\I_{\GP}$ coincides with the edge index (Sect.~\ref{subsec:6.2}).
\subsection{Bulk index and $\I_{\GP}$}
\label{subsec:6.1}
We show the bulk index coincides with the minus of $\I_{\GP}$.
The key ingredient is the cobordism invariance of the index.
\begin{proposition}\label{prop1}
	$\I_{\Bulk} = -\I_{\GP}$.
\end{proposition}

Let $\D^2_\eta$ be the closed unit disk in the complex plane whose boundary is $\SSS^1_\eta$, and let $\D^2_z$ be the closed domain of the complex plane surrounded by $\gamma$. We fix, on $\D^2_\eta$ and $\D^2_z$, spin$^c$ structures naturally induced by the spin$^c$ structure of the complex structure of the complex plane.
Let $X := (\D^2_\eta \times \D^2_z) \setminus (\SSS^1_\eta \times \gamma) \times \T$ and let $Y := \partial X = \SSS^1_\eta \times (\D^2_z \setminus \gamma) \times \T \sqcup (\D^2_\eta \setminus \SSS^1_\eta) \times \gamma \times \T$ be the boundary of the manifold with boundary $X$.
We consider the product spin$^c$ structure on $X$ and the boundary spin$^c$ structure on $Y$.
We take an extension\footnote{Take an extension of $H(\eta, t) - z$ onto $\D^2_\eta \times \D^2_z \times \T$ first, and then restrict it onto $X$.}
 $f=f(\eta, z, t) \colon X \rightarrow \End_\C(V)$ of $H(\eta, t) - z$, which is defined on $\SSS^1_{\eta} \times \gamma \times \T$.
The map $f$ defines an endomorphism of $\underline{V} = X \times V$ which we also write $f$.
We now consider restrictions of the endomorphism $f$ of $\underline{V}$ onto $\SSS^1_\eta \times (\D^2_z \setminus \gamma) \times \T$ and $(\D^2_\eta \setminus \SSS^1_\eta) \times \gamma \times \T$.
By Lemma~\ref{lem1}, each restriction defines an element $\alpha_z \in K^0_\cpt(\SSS^1_\eta \times (\D^2_z \setminus \gamma) \times \T)$ and $\alpha_\eta \in K^0_\cpt((\D^2_\eta \setminus \SSS^1_\eta) \times \gamma \times \T)$, respectively.
We first see that each element maps, through Thom isomorphisms and push-forward maps, to the bulk index (Lemma~\ref{lem4}) and $\I_{\GP}$ (Lemma~\ref{lem5}).
We then prove Proposition~\ref{prop1} by using the cobordism invariance of the index (Lemma~\ref{lem6}).
We denote by $F$ the composite of the inverse of the Thom isomorphism\footnote{Since $\SSS^1_\eta \times (\D^2_z \setminus \gamma) \times \T \subset \SSS^1_\eta \times \C \times \T$, we have a map $(\SSS^1_\eta \times \C \times \T)^+ \rightarrow (\SSS^1_\eta \times (\D^2_z \setminus \gamma) \times \T)^+$ (see also Sect. \ref{sec:2.1.1}) which induces an isomorphism between $K^0_{\cpt}(\SSS^1_\eta \times (\D^2_z \setminus \gamma) \times \T)$ and $K^0_{\cpt}(\SSS^1_\eta \times \C \times \T)$. We here identify these two by this isomorphism.}
 $\beta_z^{-1} \colon K^0_\cpt(\SSS^1_\eta \times (\D^2_z \setminus \gamma) \times \T) \rightarrow K^0(\SSS^1_\eta \times \T)$ and $\ind_{\SSS^1_\eta \times \T} \colon K^0(\SSS^1_\eta \times \T) \rightarrow \Z$.
The relation between $\alpha_z$ and $\I_\Bulk$ is stated as follows, which follows easily from Atiyah--Bott's elementary proof of the Bott periodicity theorem \cite{AB64}.
\begin{lemma}\label{lem4}
	$\beta_z^{-1}(\alpha_z) = [E_\B]$.
Thus we have $F(\alpha_z) = -\I_\Bulk$.
\end{lemma}

We denote by $F'$ the composite of the inverse of the Thom isomorphism $\beta_\eta^{-1} \colon$ $K^0_\cpt((\D^2_\eta \setminus \SSS^1_\eta) \times \gamma \times \T) \rightarrow K^0(\gamma \times \T)$ and $\ind_{\gamma \times \T} \colon K^0(\gamma \times \T) \rightarrow \Z$.
The relation between $\alpha_\eta$ and $\I_\GP$ is stated as follows.
\begin{lemma}\label{lem5}
	$\beta_\eta^{-1}(\alpha_\eta) = - [E_{\GP}] + [\underline{V}^{\oplus k}]$.
Thus we have
	$F'(\alpha_\eta) = - \I_\GP$.
\end{lemma}
\begin{proof}
We have the following commutative diagram (see Sect.~7 of \cite{At68}),
\[\xymatrix{
K^{-1}(\SSS^1_\eta \times \gamma \times \T) \ar[r]^{\partial \hspace{8mm} } \ar[d]_{\mathcal{T}}&K^0(\D^2_\eta \times \gamma \times \T, \SSS^1_\eta \times \gamma \times \T)\ar[d]^{-\beta_\eta^{-1}}\\
[\gamma \times \T, \Fred(l^2(\Z_{\geq 0}))] \ar[r]^{\mathop{\mathrm{index}}}& K^0(\gamma \times \T)\\
}\]
where $\partial$ is the boundary map of the long exact sequence for the pair $(\SSS^1_\eta \times \D^2_z \times \T, \SSS^1_\eta \times \gamma \times \T)$, and a map $\mathcal{T}$ from $K^{-1}(\SSS^1_\eta \times \gamma \times \T)$ to $[\gamma \times \T, \Fred(l^2(\Z_{\geq 0}))]$ is given by taking a family of Toeplitz operators.
Note that $-\beta_\eta$ is given by taking a cup product with the element $[\underline{\C}, \underline{\C} ; \bar{z}] \in K^0_\cpt(\D^2_\eta \setminus \SSS^1_\eta)$.
By Lemma~\ref{lem1}, we have an element $[H(\eta, t) - z] \in K^{-1}(\SSS^1_\eta \times \gamma \times \T)$.
This element maps to $\alpha_\eta \in K^0_\cpt((\D^2_\eta \setminus \SSS^1_\eta) \times \gamma \times \T)$ by the boundary map $\partial$.
On the other hand, we have $\mathcal{T}([H(\eta, t) - z]) = \{ H^\#(t) - z \}_{\gamma \times \T}$. Now since $H^\flat(z,t)$ is a finite rank perturbation of $H^\#(t) - z$, we have $\{ H^\#(t) - z \}_{\gamma \times \T} = \{ H^\flat(z,t)\}_{\gamma \times \T}$ as an element of $[\gamma \times \T, \Fred(\mathcal{H})]$, and by Lemma~\ref{lem3}
we have $\mathop{\mathrm{index}}(\{ H^\flat(z,t)\}_{\gamma \times \T}) = [E_{\GP}] - [\underline{V}^{\oplus k}]$.
By the commutativity of the above diagram, we have $-\beta_\eta^{-1}(\alpha_\eta) = [E_{\GP}] - [\underline{V}^{\oplus k}]$.
\end{proof}
\begin{Proof}{\bf of Proposition~\ref{prop1}}
We now show the Proposition~\ref{prop1}.
By Lemma~\ref{lem1}, the endomorphism $f$ defines the following element $\alpha$ of the compactly supported $K$-group,
\begin{equation*}
	\alpha := [\underline{V}, \underline{V} ; f(\eta, z, t)] \in K^0_{\cpt}(X).
\end{equation*}
$\alpha$ gives a cobordism between $\alpha_z$ and $\alpha_\eta$, that is,
let $i \colon Y \hookrightarrow X$ be the inclusion, then we have $i^*(\alpha) = (\alpha_z, \alpha_\eta)$. By Lemma~\ref{lem4} and~\ref{lem5}, we have $(F \oplus F')(\alpha_z, \alpha_\eta) = -\I_\Bulk - \I_\GP$. Thus the following lemma, which states the cobordism invariance of the index, is suffice to prove Proposition~\ref{prop1}.
\begin{lemma}\label{lem6}
$(F \oplus F')(\alpha_z, \alpha_\eta) = 0$.
\end{lemma}
\begin{proof}
By the Whitney embedding theorem, there exists a neat embedding $E \colon X \rightarrow \R^M \times [0, +\infty)$ for sufficiently large even integer $M$, such that the boundary $Y$ of $X$ maps to $\R^M \times \{0\}$.\footnote{For the proof of the Whitney embedding theorem for an embedding of a non-compact smooth manifold with boundary into a half-space, see \cite{Ad93,Hir94}, for example.
It is also easy to construct such embedding explicitly for our space $X$ (see Sect. \ref{sec:7}).} We fix a spin$^c$ structure on $\R^M \times [0, +\infty)$. Then a normal bundle of this embedding has a naturally induced spin$^c$ structure.
Let $i' \colon \R^M \hookrightarrow \R^M \times [0, +\infty)$ be an inclusion given by $i'(x) = (x, 0)$. We have the following commutative diagram.
\[\xymatrix{
K^0_\cpt(X) \ar[rr]^{i^*} \ar[d]^{E_!}& & K^0_\cpt(Y) \ar[rr]^{F \oplus F'} \ar[d]_{(E|Y)_!} & &\Z \\
K^0_\cpt(\R^M \times [0, +\infty)) \ar[rr]^{i'^*} & & K^0_\cpt(\R^M) \ar[rru]_{\beta^{-1}}
}\]
where $\beta$ is the Thom isomorphism.
We have $i^*(\alpha) = (\alpha_z, \alpha_\eta)$ and $K^0_\cpt(\R^M \times [0, +\infty)) \cong \widetilde{K}^0(\mathbb{B}^{M+1}) = \{0\}$, where $\mathbb{B}^{M+1}$ is the $(M+1)$-dimensional closed ball which is contractible.  By the commutativity of this diagram, we have $(F \oplus F')(\alpha_z, \alpha_\eta) = 0$.
\end{proof}
\end{Proof}
\subsection{$\I_{\GP}$ and edge index}
\label{subsec:6.2}
We next show that the minus of $\I_{\GP}$ coincides with the edge index.
In order to prove this, we use a localization argument for a $K$-class. We localize the support of a $K$-class near the crossing points of the spectra of edge Hamiltonians and the loop $l$ which we take in Sect.~\ref{sec:4}. We then use the excision property of the index.
\begin{proposition}\label{prop2}
	$-\I_{\GP} = \I_{\Edge}$.
\end{proposition}
We identify $V^{\oplus k}$ with the image of the projection $1 - P_{\geq k}$, that is, the closed subspace of $l^2(\Z_{\geq 0} ; V)$ consists of sequences $\varphi = \{ \varphi_n \}$ which satisfies $\varphi_n = 0$ for $n \geq k$.
For $(z, t)$ in $\gamma \times \T$, we consider the following map,
\begin{equation*}
	g_{z, t} := (H^\#(t) - z)|_{(E_{\GP})_{z, t}} \colon (E_{\GP})_{z, t} \rightarrow V^{\oplus k}.
\end{equation*}
By the definition of $(E_{\GP})_{z, t}$, the space $(H^\#(t) - z)((E_{\GP})_{z, t})$ is contained in $V^{\oplus k}$. Therefore we have a bundle homomorphism $g \colon \tilde{E}_{\GP}|_{\tilde{\T}}\rightarrow \underline{V}^{\oplus k}$.
\begin{lemma}\label{lem7}
For $(z, t) \in \gamma \times \T$, we have,
$\Ker g_{z,t} = \Ker (H^\#(t) - z)$.
\end{lemma}
\begin{proof}
$\varphi$ is an element of $\Ker g_{z,t}$ if and only if $\varphi$ is an element of $l^2(\Z_{\geq 0} ; V)$ which satisfies
$(H(t)\varphi)_n = z \varphi_n$ for $n \geq k$ (see Remark~\ref{rem6}), and
$(H(t)\varphi)_n = z \varphi_n$ for $0 \leq n \leq k-1$ (since $\varphi$ is an element of $\Ker g_{z,t}$). This is equivalent to saying that $\varphi$ is an element of $\Ker(H^\#(t)-z)$.
\end{proof}
\begin{remark}\label{rem10}
If there are no edge states, then $g$ is a bundle isomorphism, and we have a trivialization of the bundle $E_\GP$.
\end{remark}
We now deform the torus $\gamma \times \T$, so that, instead of the loop $\{ \mu \} \times \T$, the deformed torus intersects $\R \times \T \subset \C \times \T$ on the loop $l$ which we take in Sect.~\ref{sec:4}, and that, at each crossing point, a neighborhood of the point in the deformed torus is contained in $\C \times \{ t \}$ for some $t$ in $\T$. We denote this deformed torus by $\tilde{\T}$ (see Figure $2$).
Since $l$ is contained in the resolvent set of bulk Hamiltonians, such $\tilde{\T}$ is contained in the set $\{ (z,t) \in \C \times \T \ | \ z \notin \sigma(H(t)) \}$.
Thus we can consider the restriction of $\tilde{E}_{\GP}$ onto $\tilde{\T}$ (see Remark~\ref{rem7}).
$\tilde{\T}$ is deformed continuously to $\gamma \times \T$ (see Figure $2$, and consider the {\em ``crushing the ledge'' map}).
By using this deformation,
$\Tilde{\T}$ has a spin$^c$ structure naturally induced by that of $\gamma \times \T$.
By using this spin$^c$ structure, we have the map $\ind_{\tilde{\T}} \colon K^0(\tilde{\T}) \rightarrow \Z$.
\begin{lemma}\label{lem8}
	$\I_{\GP} = \ind_{\tilde{\T}}([\tilde{E}_{\GP} |_{\tilde{\T}}, \underline{V}^{\oplus k} ; g])$.
\end{lemma}
\begin{proof}
By the naturality of the index, the integer $\ind_{\gamma \times \T}([E_{\GP}] - [\underline{V}^{\oplus k}])$ coincides with $\ind_{\tilde{\T}}([\tilde{E}_{\GP} |_{\tilde{\T}}] - [\underline{V}^{\oplus k}])$. Thus, we have
\begin{equation*}
	\I_{\GP} = \ind_{\tilde{\T}}([\tilde{E}_{\GP} |_{\tilde{\T}}] - [\underline{V}^{\oplus k}]) =  \ind_{\tilde{\T}}([\tilde{E}_{\GP} |_{\tilde{\T}}, \underline{V}^{\oplus k} ;  0]) = \ind_{\tilde{\T}}([\tilde{E}_{\GP} |_{\tilde{\T}}, \underline{V}^{\oplus k} ; g]).
\end{equation*}
\end{proof}

By Lemma~\ref{lem7}, the support of the bundle homomorphism $g \colon \tilde{E}_{\GP} |_{\tilde{\T}} \rightarrow \underline{V}^{\oplus k}$ is contained in the spectrum of edge Hamiltonians, and coincides with the set of crossing points $\{ u_1, \ldots, u_m \}$.
For each crossing point $u_a$, we take an open disk neighborhood $U_a$ of $u_a$ in $\tilde{\T}$ which is also contained in $\C \times \{ t \}$ for some $t$ in $\T$.
We take $U_a$ small enough so that $U_a$ does not intersect one another.
Then we have following elements of compactly supported $K$-groups,
\begin{equation*}
	[\tilde{E}_{\GP} |_{U_a}, \underline{V}^{\oplus k} ; g] \in K^0_\cpt (U_a), \hspace{2mm} a=1, \ldots, m.
\end{equation*}
For each $a \in \{ 1, \cdots, m\}$, we consider a spin$^c$ structure on $U_a$ which is the restriction of the spin$^c$ structure of $\tilde{\T}$.
By using this spin$^c$ structure, we have a homomorphism $\ind_{U_a, \tilde{\T}} \colon K^0_\cpt (U_a) \rightarrow \Z$. We later consider another spin$^c$ structure on $U_a$, so we write subscript $(U_a, \tilde{\T})$ in order to indicate which spin$^c$ structure is used.
By the excision property of the index, we have,
\begin{equation}\label{eq1}
\ind_{\tilde{\T}}([\tilde{E}_{\GP} |_{\tilde{\T}}, \underline{V}^{\oplus k} ; g]) = \sum_{a=1}^m \ind_{U_a, \tilde{\T}}([\tilde{E}_{\GP} |_{U_a}, \underline{V}^{\oplus k} ; g]).
\end{equation}

In order to prove Proposition~\ref{prop2}, it is enough to show the following Lemma.
\begin{lemma}\label{lem9}
For each crossing point $u_a$, the following holds.
\begin{enumerate}
\renewcommand{\labelenumi}{(\arabic{enumi})}
\item If $u_a$ is contained in the interval $\mathcal{L}(\tilde{c}_i, \tilde{c}_{i+1}) \times \{ t_i \}$ where $\tilde{c}_i > \tilde{c}_{i+1}$ and if its multiplicity is $r_a$, then $\ind_{U_a, \tilde{\T}}([\tilde{E}_{\GP} |_{U_a}, \underline{V}^{\oplus k} ; g]) = + r_a$.
\item If $u_a$ is contained in the interval $\mathcal{L}(\tilde{c}_i, \tilde{c}_{i+1}) \times \{ t_i \}$ where $\tilde{c}_i < \tilde{c}_{i+1}$ and if its multiplicity is $r_a$, then $\ind_{U_a, \tilde{\T}}([\tilde{E}_{\GP} |_{U_a}, \underline{V}^{\oplus k} ; g]) = - r_a$.
\end{enumerate}
\end{lemma}
In order to prove Lemma~\ref{lem9}, we need some Lemmas.
We now consider on $U_a$ a complex structure induced by the inclusion $U_a \subset \C \times \{ t \} \cong \C$.
\begin{lemma}\label{lem10}
$\tilde{E}_{\GP}|_{U_a} \rightarrow U_a$ is a holomorphic vector bundle.
\end{lemma}
\begin{proof}
Let us consider the following operator,
$$
D(z,t) :=
\left( 
    \begin{array}{cc}
           0&H^\flat(z,t)^*\\
           H^\flat(z,t)&0
    \end{array}
\right).
$$
$D(z,t)$ is a bounded linear self-adjoint operator on $l^2(\Z_{\geq 0} ; V) \oplus l^2(\Z_{\geq 0} ; V)$.
We have $\Ker D(z,t) = \Ker   H^\flat(z,t) \oplus \Ker H^\flat(z,t)^* \cong (E_\GP)_{z, t} \oplus V^{\oplus k}$. By Lemma~\ref{lem3}, the rank of $\Ker D(z,t)$ is constant on $U_a$.
By considering the spectral decomposition of the self-adjoint operator $D(z,t)$, it is easy to see that we can take a positively oriented smooth simple closed curve $C$ in $\C$, which does not intersects with $\sigma(D(z,t))$ and contains just the zero eigenvalue inside $C$ for any $(z, t) \in U_a$.
Then Riesz projections give a holomorphic family of projections $p(z, t) := \frac{1}{2\pi i}\int_C (\lambda 1 - D(z,t))^{-1}d \lambda$ parametrized by $U_a$.
The images of this family is $\tilde{E}_{\GP}|_{U_a} \oplus \underline{V}^{\oplus k}$, and is a holomorphic vector bundle on $U_a$.
Thus $\tilde{E}_{\GP}|_{U_a}$ is a holomorphic vector bundle.
\end{proof}
\begin{lemma}\label{lem11}
Let $z$ be the complex coordinate of $U_a$, where $U_a$ is considered as a subspace of $\C$. Then $[\tilde{E}_{\GP} |_{U_a}, \underline{V}^{\oplus k} ; g]$, the element of the $K$-group $K^0_\cpt (U_a)$, is expressed as, $[\tilde{E}_{\GP} |_{U_a}, \underline{V}^{\oplus k} ; g] = [\underline{\C}, \underline{\C} ; z^{r_a}]$.
\end{lemma}
\begin{proof}
Since $\tilde{E}_{\GP}|_{U_a}$ is a holomorphic vector bundle, $g \colon \tilde{E}_{\GP} |_{U_a} \rightarrow \underline{V}^{\oplus k}$ is a holomorphic bundle map, and $\det g \colon \det E_{\GP}|_{U_a} \rightarrow \underline{\C}$ also is holomorphic.
We take a holomorphic trivialization of $\det \tilde{E}_{\GP}|_{U_a}$, then $\det g$ is holomorphic on $U_a$ which is zero at $u_a$ whose order is $r_a$. By the Taylor series, $\det g$ is expressed as $z^{r_a}h(z)$ where $h(z)$ is a nowhere vanishing function on $U_a$.
Thus we have,
\begin{equation*}
	[\tilde{E}_{\GP}|_{U_a}, \underline{V}^{\oplus k} ; g] = [\det \tilde{E}_{\GP} |_{U_a}, \underline{\C} ; \det g] = [\underline{\C}, \underline{\C} ; z^{r_a}h(z)] = [\underline{\C}, \underline{\C} ; z^{r_a}].
\end{equation*}
\end{proof}

Note that $U_a$ is considered as a subspace of two spaces. One is $\C \times \{ t \}$, and the other is $\tilde{\T}$. We have two induced spin$^c$ structures on $U_a$, and they can be different.
\begin{Proof}{\bf of Lemma \ref{lem9}}
If the assumption of ($1$) holds, these two orientations are the same, and in the case ($2$), they are the opposite.
Thus, by Lemma~\ref{lem11}, we have,
\begin{equation*}
\ind_{U_a, \tilde{\T}}([\tilde{E}_{\GP} |_{U_a}, \underline{V}^{\oplus k} ; g]) =
\left\{
\begin{aligned}
+ r_a, & \hspace{3mm} \text{if ($1$) holds,}   \\
- r_a, & \hspace{3mm} \text{if ($2$) holds.}
\end{aligned}
\right.
\end{equation*}
\end{Proof}
\begin{Proof}{\bf of Proposition~\ref{prop2}}
By Lemma~\ref{lem9} and the definition of the $\Z$-valued spectral flow (Definition~\ref{sf}), we have $\sum_{i=1}^m \ind_{U_a, \tilde{\T}}([\tilde{E}_{\GP} |_{U_a}, \underline{V}^{\oplus k} ; g]) = \mathrm{sf}(H^\#(t) - \mu)$. Thus by Lemma~\ref{lem8} and Eq. (\ref{eq1}), we have
\begin{gather*}
\I_\GP =  \ind_{\tilde{\T}}([\tilde{E}_{\GP} |_{\tilde{\T}}, \underline{V}^{\oplus k} ; g]) = \sum_{a=1}^m \ind_{U_a, \tilde{\T}}([\tilde{E}_{\GP} |_{U_a}, \underline{V}^{\oplus k} ; g])\\ = \mathrm{sf}(H^\#(t) - \mu) = -\I_{\Edge}.
\end{gather*}
\end{Proof}
\begin{Proof}{\bf of Theorem~\ref{mainthm}}
By Proposition~\ref{prop1} and Proposition~\ref{prop2}, we have
$\I_\Bulk = - \I_\GP = \I_\Edge$.
This completes the proof of Theorem~\ref{mainthm}.
\end{Proof}

\begin{remark}\label{nondegeneracy}
Although our cobordism argument follows that of Graf--Porta's, there is a little difference in the following sense.
Graf--Porta extended the Bloch bundle to some neighborhood of $\SSS^1_\eta \times \T$ and relate wave functions which make this bundle and $E_\GP$ through the eigenfunction of $H(\eta, t) - z$.
Relating these functions is not so easy in general, and so Graf--Porta assumed some non-degeneracy of the band in their proof (especially Lemma $5.8$ of \cite{GP13}).
In our proof, we treat wave functions rather indirectly.
Since the operation of taking the Fermi projection and the family index are understood $K$-theoretically (they are used in the proof of the Bott periodicity \cite{AB64,At68}), we defined elements $\alpha_z$ and $\alpha_\eta$ of $K$-groups and showed that they map to $[E_\B]$ (Lemma \ref{lem4}) and $-[E_\GP] + [\underline{V}^{\oplus k}]$ (Lemma \ref{lem5}), respectively.
Then we construct the element $\alpha$ by using $H(\eta, t) -z$ itself and gives a cobordism between $\alpha_z$ and $\alpha_\eta$.
In this way, we avoid to treat wave functions directly, and so we do not need to assume some non-degenracy of the band.
\end{remark}
\section{Type $\AII$ : 2d topological insulator}
\label{sec:7}
In this section, we give a proof of the bulk-edge correspondence for the two-dimensional quantum spin Hall system along the same lines as above. 
\subsection{Setup, bulk index, edge index and bulk-edge correspondence}
On the unit circle $\T$, we consider an involution $\tau_\T \colon \T \rightarrow \T$ induced by the complex conjugation on $\C$.
$(\T, \tau_\T)$ is an involutive space.
Let $V$ be a finite dimensional Hermitian vector space. We denote the complex dimension of $V$ by $N$.
We assume that we have an anti-linear map $\Theta \colon V \rightarrow V$ which satisfies $\Theta \Theta^* = 1$ and $\Theta^2 = -1$.\footnote{Note that $\Theta^*$ is the adjoint of the anti-linear map $\Theta$ which satisfies $\langle \Theta x, y \rangle_V = \overline{\langle x, \Theta^* y \rangle_V}$ for any $x, y \in V$.}
In this case, $N$ is even.
$\Theta$ induces an anti-linear map on $l^2(\Z; V)$ by $\{ \varphi_n\}_{n \in \Z} \mapsto \{ \Theta \varphi_n \}_{n \in \Z}$. We use the same symbol $\Theta$ for this map. Note that $\Theta$ commutes with $P_{\geq k}$ and acts on subspaces $l^2(\Z_{\geq k}; V)$ for $k \in \Z$.
We define an involution $\tau_\Theta$ on $\End_\C(V)$ by $G \mapsto \Theta G \Theta^*$.
The pair $(\End_\C(V), \tau_\Theta)$ is an involutive space.
We also write $\tau_\Theta$ for the involution on $B(l^2(\Z; V))$ induced by $\Theta$ in the same way.

We consider a bulk Hamiltonian of Definition \ref{def3.1} (self-adjointness is assumed).
We assume further that the map $H(t) \colon (\T, \tau_\T) \rightarrow (B(l^2(\Z; V)), \tau_\Theta)$ is $\Z_2$-equivariant, that is,
$H(\tau_\T(t)) = \Theta H(t)\Theta^*$.
In other words, our bulk Hamiltonian satisfies the {\em (odd) time-reversal symmetry}.
In this case, the spectrum of our Hamiltonian satisfies $\mathrm{sp}(H(t)) = \mathrm{sp}(H(\tau_\T(t)))$ for any $t \in \T$.
We assume the spectral gap condition and take a real number $\mu$ as in Sect.~\ref{sec:3}.
We also take a simple closed loop $\gamma$ in the complex plane as in Sect.~\ref{sec:3}.
We here take $\gamma$ to be reflection symmetric, $\gamma = \bar{\gamma}$.
Let $\tau_\gamma$ be the involution on $\gamma$ induced by the complex conjugation.
By the Fourier transform, we obtain a continuous family of automorphisms $\{H(\eta, t) - z\}_{\SSS_\eta^{1,1} \times \gamma \times \T}$ which satisfies the following property.
\begin{lemma}\label{equivalence}
The map $\SSS_\eta^{1,1} \times (\gamma, \tau_\gamma) \times (\T, \tau_\T) \rightarrow (\End_\C(V), \tau_\Theta)$ given  by $(\eta, z, t) \mapsto H(\eta, t) - z$ is a $\Z_2$-equivariant map.
\end{lemma}
We consider the Bloch bundle $E_\B$.
By Lemma \ref{equivalence}, $\Theta$ acts on $E_\B$ and we have the Quaternionic vector bundle $(E_\B, \Theta)$ over $\SSS_\eta^{1,1} \times (\T, \tau_\T)$.
The Bloch bundle defines an element $[(E_\B, \Theta)]$ of the $KSp$-group $KSp^0(\SSS_\eta^{1,1} \times (\T, \tau_\T))$.
We next construct a map from this $KSp$-group to $\Z_2$ as follows.
We consider the $\Z_2$-equivariant embedding $m \colon \SSS_\eta^{1,1} \times (\T, \tau_\T) \hookrightarrow \R^{1,1} \times \R^{1,1} \times \R^{2,0}$ defined by $(\eta, t) \mapsto (\eta, t, 0)$. The normal bundle of this embedding is a trivial bundle whose fiber is $\R^{2,2}$. Thus we have $m^{\AII}_! \colon KSp^0(\SSS_\eta^{1,1} \times (\T, \tau_\T)) \rightarrow KSp^0_\cpt(\R^{1,1} \times \R^{1,1} \times \R^{2,0})$.
Since $KSp^0_\cpt(\R^{1,1} \times \R^{1,1} \times \R^{2,0}) \cong \Z_2$, we have the following map.
\begin{equation*}
	\ind^{\AII}_{\SSS^1_\eta \times \T} \colon KSp^0(\SSS_\eta^{1,1} \times (\T, \tau_\T)) \rightarrow \Z_2.
\end{equation*}
An easy calculation shows that the $KSp$-group $KSp^0(\SSS_\eta^{1,1} \times (\T, \tau_\T))$ is isomorphic to $\Z \oplus \Z_2$, where the $\Z$ summand corresponds to the class of product bundles.
It is easily checked that $\ind^{\AII}_{\SSS^1_\eta \times \T}$ maps the class of product bundles to $0$.
\begin{definition}[Bulk index]
We define the {\em bulk index} of our system by
\begin{equation*}
	\I_{\Bulk}^{\AII} := \ind^{\AII}_{\SSS^1_\eta \times \T}([(E_\B, \Theta)]) \in \Z_2.
\end{equation*}
\end{definition}

We next consider the edge Hamiltonian $H^\#(t)$.
Since $\Theta$ commutes with $P_{\geq k}$,
we have a continuous family of self-adjoint Fredholm operators $\{ H^\#(t) - \mu \}_{t \in \T}$ which gives a $\Z_2$-equivariant map $(\T, \tau_\T) \rightarrow (\Fred(l^2(\Z_{\geq 0} ; V)^{s.a.}, \tau_\Theta)$.
The spectrum of our edge Hamiltonian also has the symmetry with respect to $\tau_\T$ (see also Figure $3$).

\begin{figure}
  \centering
  \includegraphics[width=84mm,clip]{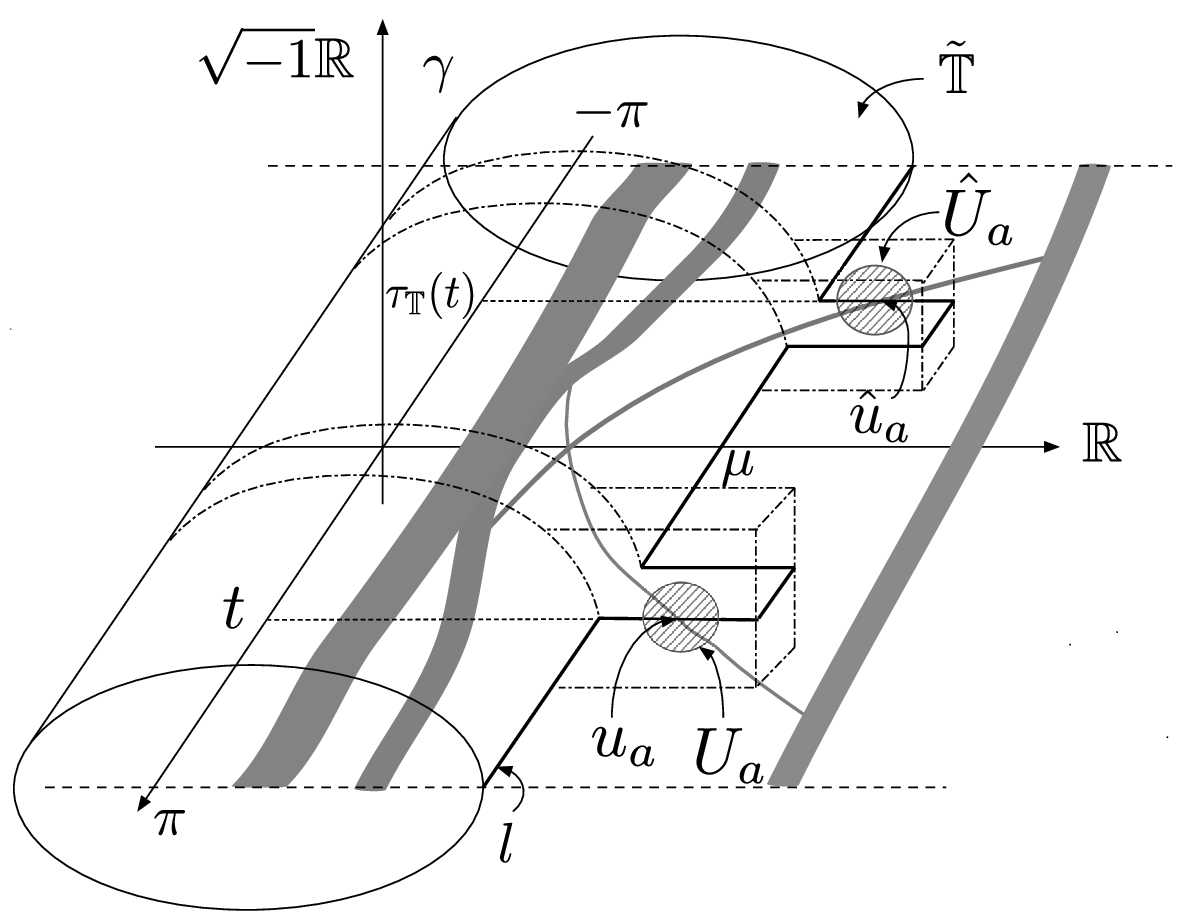}
  	\label{GPFig3}
  \caption{Under the time-reversal symmetry, the spectrum of edge Hamiltonians (gray area and gray line) are symmetric with respect to the involution $\tau_\T$ on $\T$. The loop $l$ and the deformed torus $\tilde{\T}$ are taken in such a symmetric way. Thus crossing points appear as a pair $(u_a, \hat{u}_a)$. We take neighborhoods $U_a$ and $\hat{U}_a$ in $\tilde{\T}$ of these points in such a symmetric way.}
\end{figure}

\begin{definition}[Edge index]
We define the {\em edge index} of our system by
\begin{equation*}
	\I_{\Edge}^{\AII} := \qsf(\{ H^\#(t) - \mu \}_{t \in \T}) \in \Z_2.
\end{equation*}
\end{definition}
We now revisit the definition of the $\Z_2$-valued spectral flow (Definition \ref{qsf}).
In order to define $\qsf(\{ H^\#(t) - \mu \}_{t \in \T})$, we choose $t_i$ and $c_i$.
We now consider a path $l \in \C \times [0, \pi]$ defined as in the Sect.~\ref{sec:4}.
Then $\tilde{l} = l \cup (\tau \times \tau_\T)(l)$ gives a loop in $\C \times \T$.
By using this loop $\tilde{l}$, we consider a deformed torus $\tilde{\T}$ as in Sect.~\ref{subsec:6.2} in such a way that $\tilde{\T}$ is closed under the involution.
$\qsf(\{ H^\#(t) - \mu \}_{t \in \T})$ is defined by counting the number of crossing points $\{ u_1, \cdots, u_m \}$ of the spectrum of the edge Hamiltonian and the path $l$ with multiplicity, and taking mod $2$.
If we consider the path $(\tau \times \tau_\T)(l)$, there also are crossing points $\{ \hat{u}_1, \cdots, \hat{u}_m \}$ in such a symmetric way with respect to the involution (see Figure $3$).
Thus we can say that the $\Z_2$-valued spectral flow $\qsf(\{ H^\#(t) - \mu \}_{t \in \T})$ is defined by counting the {\em pair} of crossing points $(u_a, \hat{u}_a)$ with multiplicity and taking mod $2$.
We take this point of view in Proposition \ref{prop2AII}.

The following is the main theorem of this subsection, the bulk-edge correspondence for two-dimensional type $\AII$ topological insulators.
\begin{theorem}[Bulk-edge correspondence for quantum spin Hall systems]\label{2dTypeAII}
The bulk index coincides with the edge index. That is,
$\I_{\Bulk}^{\AII} = \I_{\Edge}^{\AII}$.
\end{theorem}
\begin{remark}
For a two-dimensional quantum spin Hall system, our $\Z_2$-valued bulk index coincides with the Kane-Mele invariant \cite{KM05b} and its extension obtained by Katsura-Koma \cite{KK16}.
If we assume an extra symmetry, we can compute the $\Z_2$-invariant as the parity of the spin-Chern number \cite{SB15}.
\end{remark}
\subsection{Proof of the bulk-edge correspondence}
As in the case of type A systems, we consider decaying solutions in order to prove Theorem \ref{2dTypeAII}.
By Lemma \ref{lem2} and Lemma \ref{lem3}, we can define $E_{\GP}$ as the kernels of a family of Fredholm operators.
Since our Hamiltonian satisfies time-reversal symmetry and the operator $\Theta$ acts on $l^2(\Z; V)$ point-wise, $\Theta$ acts on the bundle $E_{\GP}$.
Thus the pair $(E_{\GP}, \Theta)$ is a Quaternionic vector bundle over the involutive space $(\gamma, \tau_\gamma) \times (\T, \tau_\T)$.
Note that the cokernels $\underline{V}^{\oplus k}$ also has a Quaternionic vector bundle structure over $(\gamma, \tau_\gamma) \times (\T, \tau_\T)$  in a natural way.
Thus $(E_{\GP}, \Theta)$ defines an element of $KSp^0((\gamma, \tau_\gamma) \times (\T, \tau_\T))$.
We now construct a map from $KSp^0((\gamma, \tau_\gamma) \times (\T, \tau_\T))$ to $\Z_2$ as follows.
We consider an $\Z_2$-equivariant embedding $m' \colon (\gamma, \tau_\gamma) \times (\T, \tau_\T) \hookrightarrow \R^{1,1} \times \R^{1,1} \times \R^{2,0}$ defined by $(z, t) \mapsto (z, t, 0)$. The normal bundle of this embedding is a trivial bundle whose fiber is $\R^{2,2}$. Thus we have, ${m'}^{\AII}_! \colon KSp^0((\gamma, \tau_\gamma) \times (\T, \tau_\T)) \rightarrow KSp^0_\cpt((\gamma, \tau_\gamma) \times (\T, \tau_\T) \times \R^{2,2})$.
Since $KSp^0_\cpt(\R^{1,1} \times \R^{1,1} \times \R^{2,0}) \cong \Z_2$, we have
\begin{equation*}
	\ind^{\AII}_{\gamma \times \T} \colon KSp^0((\gamma, \tau_\gamma) \times (\T, \tau_\T)) \rightarrow \Z_2,
\end{equation*}
defined by the composition of these maps.
\begin{definition}[Graf--Porta \cite{GP13}]
Following Graf--Porta's idea, we introduce another invariant of our system by,
\begin{equation*}
	\I_{\GP}^{\AII} := \ind^{\AII}_{\gamma \times \T}([(E_{\GP}, \Theta)]) \in \Z_2.
\end{equation*}
\end{definition}
As in the type A case, the proof of Theorem \ref{2dTypeAII} is divided into two parts.
We first show that the bulk index coincides with $\I_{\GP}^{\AII}$ by using the cobordism invariance of the index (Proposition \ref{prop1AII}).
We next show that $\I_{\GP}^{\AII}$ equals to the edge index by the localization of the $KSp$-class (Proposition \ref{prop2AII}).
\begin{proposition}\label{prop1AII}
$\I_{\Bulk}^{\AII} = \I_{\GP}^{\AII}.$
\end{proposition}
We consider the space $X$ and $Y$ considered in the case of type A systems with involutions.
Let $(X, \tau_X) := (\D^{1,1}_\eta \times \D^{1,1}_z) \setminus (\SSS^{1,1}_\eta \times (\gamma, \tau_\gamma)) \times (\T, \tau_\T)$
and $(Y, \tau_Y) := \SSS^{1,1}_\eta \times (\D^{1,1}_z \setminus (\gamma, \tau_\gamma)) \times (\T, \tau_\T) \sqcup (\D^{1,1}_\eta \setminus \SSS^{1,1}_\eta) \times (\gamma, \tau_\gamma) \times (\T, \tau_\T)$.
Take an extension of $H(\eta, t) - z$ as in the case of type A, and we obtain elements
\begin{equation*}
	\alpha^{\AII}_z \in KSp^0_\cpt(\SSS^{1,1}_\eta \times (\D^{1,1}_z \setminus (\gamma, \tau_\gamma)) \times (\T, \tau_\T)),
\end{equation*}
\begin{equation*}
	\alpha^{\AII}_\eta \in KSp^0_\cpt((\D^{1,1}_\eta \setminus \SSS^{1,1}_\eta) \times (\gamma, \tau_\gamma) \times (\T, \tau_\T)).
\end{equation*}
Let $F_{\AII}$ be the composition of the inverse of the Bott periodicity isomorphism
\begin{equation*}
(\beta^{\AII}_z)^{-1} \colon KSp^0_\cpt(\SSS^{1,1}_\eta \times (\D^{1,1}_z \setminus (\gamma, \tau_\gamma)) \times (\T, \tau_\T))
	\rightarrow
		KSp^0(\SSS^{1,1}_\eta \times (\T, \tau_\T)),
\end{equation*}
and $\ind^{\AII}_{\SSS^1_\eta \times \T} \colon KSp^0(\SSS^{1,1}_\eta \times (\T, \tau_\T)) \rightarrow \Z_2$.
The relation between $\alpha^{\AII}_z$ and $\I^{\AII}_\Bulk$ is stated as follows, which follows easily from Atiyah's proof of the Bott periodicity theorem for $KR$-theory \cite{At66}.
\begin{lemma}\label{AIIlem4}
	$(\beta^{\AII}_z)^{-1}\alpha^{\AII}_z = [(E_\B, \Theta)]$.
Thus we have
	$F_{\AII}(\alpha^{\AII}_z) = \I^{\AII}_\Bulk$.
\end{lemma}
Let $F'_{\AII}$ be the composition of the inverse of the Bott periodicity isomorphism,
\begin{equation*}
(\beta^{\AII}_\eta)^{-1} \colon KSp^0_\cpt((\D^{1,1}_\eta \setminus \SSS^{1,1}_\eta) \times (\gamma, \tau_\gamma) \times (\T, \tau_\T)),
	\rightarrow
		KSp^0((\gamma, \tau_\gamma) \times (\T, \tau_\T))
\end{equation*}
and
$\ind^{\AII}_{\gamma \times \T} \colon KSp^0((\gamma, \tau_\gamma) \times (\T, \tau_\T)) \rightarrow \Z_2$.
The relation between $\alpha^{\AII}_\eta$ and $\I^{\AII}_\GP$ is stated as follows.
\begin{lemma}\label{AIIlem5}
	$(\beta^{\AII}_\eta)^{-1}\alpha_\eta =[(E_{\GP}, \Theta)] + [(\underline{V}^{\oplus k}, \Theta)]$.
Thus we have
	$F'_{\AII}(\alpha^{\AII}_\eta) = \I^{\AII}_\GP$.
\end{lemma}
\begin{proof}
As in Lemma \ref{lem5}, the proof is given by considering the following commutative diagram \cite{At66,At68}.
\[\xymatrix
@!C=165pt
{
KSp^{-1}(\SSS^{1,1}_\eta \times(\gamma, \tau_\gamma) \times (\T, \tau_\T))
	\ar[r]^{\partial \hspace{5mm}}
		\ar[d]_{\mathcal{T}} & KSp^0((\D^{1,1}_\eta, \SSS^{1,1}_\eta) \times (\gamma, \tau_\gamma) \times (\T, \tau_\T)) \ar[d]^{(\beta^{\AII}_\eta)^{-1}}\\
[(\gamma, \tau_\gamma) \times (\T, \tau_\T), (\Fred(l^2(\Z_{\geq 0}; \C^2), \tau_{\Theta})] \ar[r]^{\hspace{1cm} \mathop{\mathrm{index}}}& KSp^0((\gamma, \tau_\gamma) \times (\T, \tau_\T)).\\
}\]
The map index is given by a family index as explained in section $2.1.2$.
\end{proof}
\begin{lemma}\label{AIIlem6}
$(F_{\AII} \oplus F'_{\AII})(\alpha^{\AII}_z, \alpha^{\AII}_\eta) = 0$.
\end{lemma}
\begin{proof}
We first construct an explicit $\Z_2$-equivariant embedding $E^{\AII} \colon X \rightarrow \R^{6,4} \times ([0, +\infty), \id)$ which maps the boundary to the boundary and whose normal bundle is a trivial bundle of fiber $\R^{3,3}$.
Let $\theta \colon (\D^{1,1}_\eta \times \D^{1,1}_z) \setminus (\SSS^{1,1}_\eta \times (\gamma, \tau_\gamma)) \rightarrow \R^{0,1}$ be a function given by
$\theta(\eta, z) = \mathrm{arg}((1-|\eta|) + \sqrt{-1}(1-|z|))$, where $\arg(x+\sqrt{-1}y)$ is the argument of $x+\sqrt{-1}y$.
We here consider $\D^{1,1}_\eta$ and $\D^{1,1}_z$ as a unit disc in the complex plane and $|\cdot|$ is the absolute value in the complex plane.
By using this, we define a map $f \colon (\D^{1,1}_\eta \times \D^{1,1}_z) \setminus (\SSS^{1,1}_\eta \times (\gamma, \tau_\gamma)) \rightarrow \R^{0,1} \times ([0, +\infty), \id)$ by $f(\eta, z) = (\cos(2\theta(\eta, z)), \sin(2\theta(\eta,z)))$, where $\R^{0,1} \times ([0, +\infty), \id)$ is identified with the upper half-plane in $\R^{0,2}$.
$\theta$ and $f$ are $\Z_2$-equivariant maps.
Let $e$ be the map
\begin{gather*}
	e \colon (X, \tau_X)= (\D^{1,1}_\eta \times \D^{1,1}_z) \setminus (\SSS^{1,1}_\eta \times (\gamma, \tau_\gamma)) \times (\T, \tau_\T) \hspace{3cm}\\
		\hspace{2cm} \longrightarrow \R^{1,1} \times \R^{1,1} \times (\R^{0,1} \times ([0, +\infty), \id)) \times \R^{1,1} = \R^{3,4} \times ([0, +\infty), \id),
\end{gather*}
given by
$e(\eta, z, t) = (\eta, z, f(\eta, z), t)$.
The map $e$ is a $\Z_2$-equivariant neat embedding which maps the boundary to the boundary whose normal bundle is a trivial bundle of fiber $\R^{0,3}$.
By using this, we define an embedding $E^{\AII} \colon (X, \tau_X) \rightarrow (\R^{3,4} \times ([0, +\infty), \id)) \times \R^{3,0}$ by $E^{\AII}(\eta, z, t) = (e(\eta, z, t), 0)$.
This map $E^{\AII}$ has the desired properties.
Let $i' \colon \R^{6,4} \hookrightarrow \R^{6,4} \times ([0, +\infty), \id)$ be an inclusion given by $i'(x) = (x, 0)$. We have the following commutative diagram.
\[\xymatrix{
KSp^0_\cpt(X, \tau_X) \ar[rr]^{i^*} \ar[d]^{(E^{\AII})_!}& & KSp^0_\cpt(Y, \tau_Y) \ar[rr]^{\hspace{7mm} F_{\AII} \oplus F'_{\AII}} \ar[d]_{(E^{\AII}|Y)_!} & &\Z_2 \\
KSp^0_\cpt(\R^{6,4} \times ([0, +\infty), \id)) \ar[rr]^{i'^*} & & KSp^0_\cpt(\R^{6,4})  \ar[rru]_{(\beta^{\AII})^{-1}} & &
}\]
We have $i^*(\alpha^{\AII}) = (\alpha^{\AII}_z, \alpha^{\AII}_\eta)$ and $KSp^0_\cpt(\R^{6,4} \times ([0, +\infty), \id)) = 0$ since the one point compactification of $\R^{6,4} \times ([0, +\infty), \id)$ is a $11$ dimensional closed ball with some involution, which is $\Z_2$-equivariant contractible.
By the commutativity of this diagram, we have $(F_{\AII} \oplus F'_{\AII})(\alpha^{\AII}_z, \alpha^{\AII}_\eta) = 0$.
\end{proof}
\begin{Proof}{\bf of Proposition~\ref{prop1AII}}
By Lemma~\ref{AIIlem4}, Lemma~\ref{AIIlem5} and Lemma~\ref{AIIlem6}, we have
$0 = (F_{\AII} \oplus F'_{\AII})(\alpha^{\AII}_z, \alpha^{\AII}_\eta) = F_{\AII}(\alpha^{\AII}_z) + F'_{\AII}(\alpha^{\AII}_\eta) = \I^{\AII}_\Bulk + \I^{\AII}_\GP$.
\end{Proof}
\begin{proposition}\label{prop2AII}
$\I_{\GP}^{\AII}=\I_{\Edge}^{\AII}$.
\end{proposition}
This Proposition follows by the localization of the $KSp$-class by using the Dirichlet boundary condition as in the case of type A (Proposition \ref{prop2}).
We take a deformed torus $\tilde{\T}$ such that $\tilde{\T}$ is closed under the involution of $(\C, \tau) \times \R^{1,1}$.
We write $\tau_{\tilde{\T}}$ for the involution on $\tilde{\T}$ induced by this involution.
Then we define a map $\ind_{\tilde{\T}}^{\AII} \colon KSp^0(\tilde{\T}, \tau_{\tilde{\T}}) \rightarrow \Z_2$ in the same way.
The bundle homomorphism $g \colon (\tilde{E}_{\GP}|_{\tilde{\T}}, \Theta) \rightarrow (\underline{V}^{\oplus k}, \Theta)$ is a $\Z_2$-equivariant homomorphism.
The following statement holds which corresponds to Lemma \ref{lem8}.
\begin{lemma}\label{AIIlem8}
$\I_{\GP}^{\AII} = \ind_{\tilde{\T}}^{\AII}([(\tilde{E}_{\GP}|_{\tilde{\T}}, \Theta), (\underline{V}^{\oplus k}, \Theta); g])$.
\end{lemma}
The support of $g$ corresponds to crossing points. By using excision property, we obtain the following equation which corresponds to Eq. (\ref{eq1}).
\begin{gather}
\ind_{\tilde{\T}}^{\AII}([(\tilde{E}_{\GP}|_{\tilde{\T}}, \Theta), (\underline{V}^{\oplus k}\hspace{-0.5mm}, \Theta); g]) \hspace{4cm} \nonumber
	\\
	\hspace{1cm}=	\sum_{a=1}^{m} \ind_{U_a \cup \hat{U}_a}^{\AII}\hspace{-0.5mm}([(\tilde{E}_{\GP}|_{U_a \cup \hat{U}_a}\hspace{-0.5mm}, \Theta), (\underline{V}^{\oplus k}\hspace{-0.5mm}, \Theta); g]).\label{eq1AII}
\end{gather}
Now the problem is how to count the contribution of each crossing point to $\I_{\GP}^{\AII}$ by using $KSp$-theory. 
In the type $\AII$ case, we have the following Lemma which corresponds to Lemma \ref{lem9}.
Note that since the spectrum of edge Hamiltonians (and our choice of loop $l$ and deformed torus $\tilde{\T}$) are symmetric with respect to the involution on $\T$, each crossing point appears as a pair.
\begin{figure}
  \centering
  \includegraphics[width=110mm,clip]{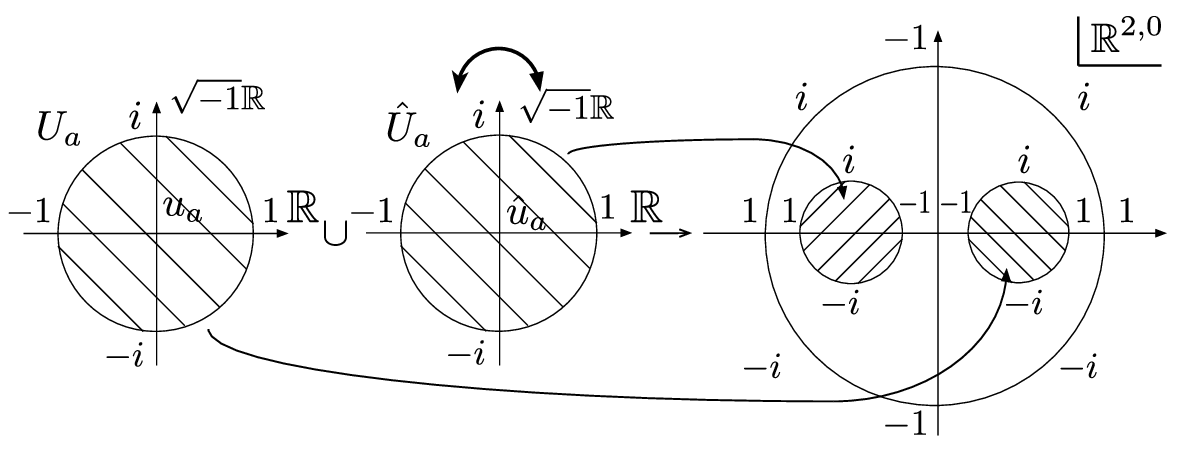}
  \caption{The embedding $U_a \cup \hat{U}_a \rightarrow \R^{2,0}$ is indicated at this figure. An example of the value of the determinant of the clutching function and that of the one on $\SSS^{2,0}$ is indicated by the numbers $1, i, -1, -i$. We here consider the case $r_a =1$}
\end{figure}
\begin{lemma}\label{AIIlem9}
For each pair of crossing points $u_a$ and $\hat{u}_a$, the following holds.
If the multiplicity of the crossing point $u_a$ is $r_a$, then
$$\ind_{U_a \cup \hat{U}_a}^{\AII}([(\tilde{E}_{\GP}|_{U_a \cup \hat{U}_a}, \Theta), (\underline{V}^{\oplus k}, \Theta); g]) = r_a \mod 2.$$
\end{lemma}
\begin{proof}
At each pair of crossing points $u_a$ and $\hat{u}_a$, we take neighborhoods $U_a$ and $\hat{U}_a$ in $\tilde{\T}$ of these points such that they maps to each other by the involution on the deformed torus $\tilde{\T}$.
We then embed the involutive space $U_a \cup \hat{U}_a$ into $\D^{2,0}$ in a $\Z_2$-equivariant way.
More explicitly, we take two small balls in the unit ball in $\D^{2,0}$ in such a way that they maps to each other by the involution on $\R^{2,0}$.
We then consider such an embedding as in Fig.~$4$. We embed $U_a$ into the right ball in $\D^{2,0}$ and $\hat{U}_a$ by first flipping $\hat{U}_a$ around the imaginary axis and embedding into the left ball in $\D^{2,0}$.
By using this embedding, we obtain an homomorphism,
$KSp^0_\cpt(U_a \cup \hat{U}_a, \tau_{\tilde{\T}}) \rightarrow KSp^0(\D^{2,0}, \SSS^{2,0})$.
The (complex) determinant of the clutching function of a representative of $\mathrm{ext}([(\tilde{E}_{\GP}|_{U_a \cup \hat{U}_a}, \Theta), (\underline{V}^{\oplus k}, \Theta); g])$ is $\Z_2$-equivariantly homotopic to $g_{r_a}$ considered in Example \ref{localizationexample}.
Thus by using  the isomorphism $KSp^0(\D^{2,0}, \SSS^{2,0}) \cong \Z_2$ of Proposition \ref{isom},
we have,
\begin{equation*}
\ind_{U_a \cup \hat{U}_a}^{\AII}([(\tilde{E}_{\GP}|_{U_a \cup \hat{U}_a}, \Theta), (\underline{V}^{\oplus k}, \Theta); g])
	= w(g_{r_a})
		=r_a \mod 2.
\end{equation*}
\end{proof}
\begin{Proof}{\bf of Proposition~\ref{prop2AII}}
By Definition \ref{qsf}, Lemma~\ref{AIIlem8}, Lemma~\ref{AIIlem9} and Eq.~(\ref{eq1AII}), we have,
\begin{gather*}
\I_{\GP}^{\AII}
	=
		\ind_{\tilde{\T}}^{\AII}([(\tilde{E}_{\GP}|_{\tilde{\T}}, \Theta), (\underline{V}^{\oplus k}, \Theta); g])\hspace{3cm} \\
	=
		\sum_{a=1}^{m} \ind_{U_a \cup \hat{U}_a}^{\AII}\hspace{-0.5mm}([(\tilde{E}_{\GP}|_{U_a \cup \hat{U}_a}\hspace{-0.5mm}, \Theta), (\underline{V}^{\oplus k}\hspace{-0.5mm}, \Theta); g])
	=
		\qsf(\{ H^\#(t) - \mu \}_{t \in \T}) = \I_{\Edge}^{\AII}.
\end{gather*}
\end{Proof}
\begin{Proof}{\bf of Theorem~\ref{2dTypeAII}}
By proposition~\ref{prop1AII} and Proposition~\ref{prop2AII}, we have
$\I_\Bulk^{\AII} = \I_\GP^{\AII} = \I_\Edge^{\AII}$.
This completes the proof of Theorem~\ref{2dTypeAII}.
\end{Proof}
\subsection*{Acknowledgements}
This work is part of a Ph.D. thesis, defended at the University of Tokyo in 2017.
The author expresses gratitude to his supervisor Mikio Furuta for his encouragement and for helpful comments and suggestions. In particular, it is pointed out to the author by him the relationship between the bulk-edge correspondence and Atiyah-Bott's elementary proof of the Bott periodicity theorem from the very beginning of this work, and the idea to use the cobordism invariance of the index in the proof of the bulk-edge correspondence is due to him.
This work started from author's collaborative research with Mikio Furuta, Motoko Kotani, Yosuke Kubota, Shinichiroh Matsuo and Koji Sato. He would like to thank them for many stimulating discussions and encouragements.
This work was supported by JSPS KAKENHI Grant Number JP17H06461.


\end{document}